\documentclass{article}

\emergencystretch=\maxdimen
\usepackage{braket}
\usepackage{pgfplots}
\pgfplotsset{compat=1.18}
\usepackage{color}
\usepackage[utf8]{inputenc} 
\usepackage{mathtools,amsmath,amsfonts,amsthm,amssymb} 
\usepackage{authblk} 
\usepackage{indentfirst} 
\usepackage[margin=3cm]{geometry} 
\usepackage{tikz-cd}
\usepackage{tocloft}
\usepackage[normalem]{ulem}
\usepackage{hyperref}
\hypersetup{
    colorlinks,
    citecolor=black,
    filecolor=black,
    linkcolor=black,
    urlcolor=black
}

\newtheorem{thm}{Theorem}
\newtheorem{lem}[thm]{Lemma}
\newtheorem{prop}[thm]{Proposition}
\newtheorem{cor}[thm]{Corollary}
\newtheorem{defn}[thm]{Definition}
\newtheorem{exmp}[thm]{Example}
\newtheorem{remark}[thm]{Remark}

\newtheorem{quest}[thm]{Question}

\newcommand{\ZZ}{\mathbb Z}

\newcommand{\Ext}{\mathrm{Ext}}
\newcommand{\Tor}{\mathrm{Tor}}

\DeclareMathOperator{\Hom}{Hom}
\DeclareMathOperator{\im}{im}

\AtEndDocument{%
  \par
  \medskip
  \begin{tabular}{@{}l@{}}%
    *\textsc{University of Warsaw}\\
    \textit{E-mail address}: \texttt{bruba@fuw.edu.pl}
  \end{tabular}
  \par
  \medskip
  \begin{tabular}{@{}l@{}}%
    **\textsc{Harvard University}\\
    \textit{E-mail address}: \texttt{byyang@alumni.caltech.edu}
  \end{tabular}}

\title{Witt Groups and Bulk-Boundary \\Correspondence for Stabilizer States}

\author{Błażej Ruba* and Bowen Yang**}

\usepackage[maxbibnames=99]{biblatex}
\begin{document}
\maketitle

\begin{abstract}
We establish a bulk--boundary correspondence for translation-invariant stabilizer states in arbitrary spatial dimension, formulated in the framework of modules over Laurent polynomial rings. To each stabilizer state restricted to half-space geometry we associate a boundary operator module. Boundary operator modules provide examples of quasi-symplectic modules, which are objects of independent mathematical interest. In their study, we use ideas from algebraic $L$-theory in a setting involving non-projective modules and non-unimodular forms. Our results about quasi-symplectic modules in one spatial dimension allow us to resolve the conjecture that every stabilizer state in two spatial dimensions is characterized by a corresponding abelian anyon model with gappable boundary. Our techniques are also applicable beyond two dimensions, such as in the study of fractons.
\end{abstract}

\tableofcontents
    \section{Introduction}

   Stabilizer states are quantum many-body states defined as the common invariant vectors of abelian subgroups of Pauli groups. These groups are discrete analogues of Heisenberg groups, with the simplest example generated by the Pauli matrices. They also arise naturally as ground states of certain local Hamiltonians with commuting terms. In many respects, they play a role analogous to that of gaussian (or~quasi-free) states for bosons and fermions. Interest in stabilizer states first arose in quantum information theory \cite{CRSS,Gottesman}. This line of work led to the discovery of integrable lattice spin models with stabilizer ground states and other rich structure. The most celebrated and influential example is Kitaev’s toric code \cite{KITAEV2003}, whose ground states realize the topological quantum field theory (TQFT) of $\mathbb{Z}_2$ gauge theory. More recently, stabilizer constructions of other TQFTs have been identified \cite{ellison2022pauli, ellison2023pauli, shirley2022three}.

    In three dimensions and beyond, stabilizer states give rise to an even richer variety of models. Alongside higher-dimensional generalizations of models realizing TQFTs, there also exist ``fracton'' models~\cite{nandkishore2019fractons, pretko2020fracton}, characterized by striking physical properties that defy conventional TQFT descriptions. Examples include Haah’s cubic code and the X-cube model~\cite{haah2013commuting, vijay2016fracton}. The discovery of such systems has motivated the development of novel algebraic frameworks for their classification and analysis~\cite{haah2013commuting, ruba2024homological, wickenden2024planon}.

    As we review in Section \ref{sec:Laurent}, the study of translation-invariant stabilizer states reduces to the following algebraic setup, which is also of independent mathematical interest. Consider the Laurent polynomial ring $R = \mathbb Z_n[x_1^{\pm 1}, \dots, x_d^{\pm 1}]$ equipped with the involution $\overline x_i = x_i^{-1}$. Let $L$ be a free $R$-module and $P =L_0^{\vphantom{*}} \oplus L_0^*$ be equipped with the sesquilinear pairing represented by the matrix
\begin{equation}
\begin{pmatrix}
0 & I_q \\
-I_q & 0
\end{pmatrix}.
\end{equation}
While in this paper we work with Laurent polynomial rings $R$, it is natural to study general rings with involution.

The central question is to what extent we can understand and classify Lagrangian submodules, namely, submodules $L \subset P$ satisfying $L^\perp = L$. Importantly, we do not assume $L$ to be a direct summand, so $L$ need not be projective. Lagrangian submodules of $P$ always exist, as~evidenced by $L_0$.

In the mathematical literature, attention is often restricted to Lagrangian submodules that are direct summands of $P$. In this setting, the pair $(L_0, L)$ defines a so-called formation~\cite{ranicki1973algebraicI,luck2024surgery}, and the classification problem falls within the scope of algebraic $L$-theory; specifically, the group $L_3(R)$. If $L$ is not a direct summand, other tools are required. The existing results are most complete for $d=2$ and prime number $n$, for which a satisfactory classification was obtained in \cite{haah2021classification}. However, even departure from the case of prime $n$ leads to a much richer landscape of stabilizer states, and requires different methods of study.


From the physical perspective, one crucial phenomenon is that stabilizer ground states can admit local excitations that cannot be created by acting on the ground state with local operators. Algebraically, they correspond to involution-antilinear homomorphisms $ L \to R$ that do not extend to $P$. In \cite{ruba2024homological}, we showed that their classification can be phrased in terms of homological algebra, namely via the module $Q^0:= \Ext^1_R(\overline{P/L},R)$. Building on this observation, we proposed considering charge modules $Q^i := \Ext^{i+1}_R(\overline{P/L},R)$, and studied their algebraic properties as well as physical interpretation in terms of excitations and symmetries. All charge modules vanish if $L$ is a direct summand. 

We remark that the dual of the charge module $Q^i$ can be computed as 
\begin{equation}
    D_i := \Tor_{i+1}^R(P/L,\widehat R) \cong \Hom_{\mathbb Z_n}(Q^i,\mathbb Z_n),
\end{equation}
as we explain in Section \ref{sec:23}. The case $i=0$ is the easiest to understand, and we find that $D_0 $ can be described in terms of certain non-local operators which detect non-trivial excitations but commute with all local operators. We believe that for stabilizer states this is the right general definition of the set of detectors proposed in \cite{wickenden2025excitation}.



The main topic of this article is a construction that we call the bulk-boundary correspondence, following standard physical terminology. It associates a~Lagrangian submodule $L$ with another module $P_\partial$, defined over the Laurent polynomial ring $R_\partial = \mathbb Z_n[x_1^{\pm 1}, \dots, x_{d-1}^{\pm 1}]$ in one fewer variable. The boundary module $P_\partial$ naturally carries a sesquilinear pairing~$\Omega_\partial$, which encodes the commutation relations of certain operators localized near the boundary of a~half-space in $\mathbb Z^d$. Related constructions are also considered in the physics literature~\cite{liang2024operator, schuster2023holographic, kobayashi2024generalized}, and find applications in quantum error correction~\cite{liang2025generalized, liang2025planar}.



Besides describing boundary conditions for lattice models in half-space geometry, our construction also leads to new results in the classification of stabilizer states. Moreover, it naturally leads to a class of objects called quasi-symplectic modules, first defined in \cite{ruba2024homological}. In their study, we define a new Witt group that generalizes $L_2(R_\partial)$, but differs from it in two key aspects: the modules are not necessarily projective, and the forms are not required to be unimodular. In fact, the failure of unimodularity of $P_\partial$ is related to the charge module $Q^0$ of $L$ via a homomorphism we call boundary-to-bulk map. Most of Section \ref{sec:bb_map} discusses conditions under which the boundary-to-bulk map is an isomorphism between $Q^0$ and an $R_\partial$-module $E_{P_\partial}$, which is constructed entirely in terms of $P_\partial$.


 Our construction has a counterpart in algebraic $L$-theory in the case where $L\subset P$ is a direct summand. Observe that $R = R_\partial[x_d^{\pm 1}]$. The bulk-boundary correspondence manifests as the fundamental theorem of $L$-theory, which relates $L_{i+1}(R)$ to $L_i(R_\partial)$~\cite{ranicki1973algebraicI, ranicki1973algebraicII}. This has been used in \cite{haah2023nontrivial,haah2025topological} to analyze Clifford quantum cellular automata (QCAs). An analogous fundamental theorem also exists in algebraic $K$-theory~\cite{quillen2006higher}. See Section 4.3 for further discussion of bulk-boundary correspondence in the case of Lagrangian direct summands. 
 

The bulk-boundary construction is applicable in any dimension $d$, including the fracton models. This is noteworthy because it contributes to systematic understanding of fractons, which is far from complete. We point out an interesting feature: the bulk-boundary map is surjective (likely an isomorphism) if the system is divided into half-spaces along a generic direction. Unlike in systems with fully mobile excitations, this surjectivity statement may fail for certain special boundary directions, namely those for which some local excitations can not be moved between the half-spaces by acting with local operators. 

In the case of two-dimensional models, we use bulk-boundary correspondence to settle a~conjecture \cite{ellison2022pauli} relating stabilizer states to a~class of TQFTs: abelian anyon models that admit topological boundary conditions. Recall that an~abelian anyon model is specified by the data of a finite abelian group $Q$ equipped with a~non-degenerate quadratic form $\theta$ (see~Appendix~\ref{app:metric_group}). In the physics literature, the elements of $Q$ are often called anyon types, the addition in $Q$ fusion, and $\theta$ the topological spin function. The existence of topological boundary conditions is equivalent to the existence of a Lagrangian subgroup in $Q$. 

In order to relate two-dimensional stabilizer states with abelian anyon models, we identify $Q$ as $Q^0$ introduced earlier (while higher $Q^i$ vanish for $d=2$ \cite{ruba2024homological}); the construction of $\theta$ goes back to \cite{levin2003fermions,haah2021classification}. Our contributions, in this article and in~\cite{ruba2024homological}, are to rigorously establish the following:
\begin{enumerate}
    \item The finiteness of $Q$, first shown by the present authors in~\cite{ruba2024homological}, with an alternative proof provided in this paper.
    \item The non-degeneracy of $\theta$, proved in this paper using the bulk--boundary correspondence together with a corresponding result on quasi-symplectic modules in dimension one. 
    \item The existence of a Lagrangian subgroup of $(Q, \theta)$, obtained via the bulk--boundary correspondence and our theory of Witt groups of quasi-symplectic modules in dimension one.
    \item That a choice of boundary conditions for a stabilizer state in half-space geometry gives rise to a Lagrangian subgroup of $Q$. 
\end{enumerate}

The article is organized as follows. In Section~\ref{sec:Laurent}, we introduce the algebraic framework underlying translation-invariant stabilizer codes, modeling Pauli operators and their commutation relations using modules over Laurent polynomial rings. We review the definition of the charge module \( Q_L \), describe its homological description, and discuss the concept of coarse-graining. In Section~\ref{sec:qs}, we introduce quasi-symplectic modules as a generalization of standard symplectic modules. Section~3.2 develops a Witt group that classifies such modules up to metabolic equivalence. Section~3.3 specializes to the one-dimensional case, in which we associate to $P_\partial$ a finite abelian group with a quadratic form, and relate this construction to the Witt group. In Section~4, we formulate a bulk-boundary correspondence for stabilizer codes defined on half-spaces. Section~4.1 introduces the boundary operator module \( P_\partial \) and shows that it is quasi-symplectic. Section~4.2 presents a map between $Q^0$ and a module $E_{P_{\partial}}$ naturally associated to $P_\partial$; we discuss conditions for its injectivity and surjectivity, and prove that it is always an isomorphism for $d=2$. Section~4.3 compares this construction to the boundary algebra associated with Clifford QCAs. Appendix \ref{app:Heisenberg} reviews background material on Heisenberg type groups, and Appendix \ref{app:metric_group} on quadratic forms and metric groups.

\section*{Acknowledgments}
We are grateful to Yu-An Chen, Dan Freed, Wilbur Shirley, Shmuel Weinberger, and Evan Wickenden for valuable discussions. We also thank Dan Freed, Mike Freedman and Anton Kapustin for reading a draft of this paper. The work of B.R. was supported by the National Science Centre (NCN) under Sonata Bis 13 grant no. 2023/50/E/ST1/00439. B.Y. acknowledges support from the Harvard CMSA and from the Simons Foundation through the Simons Collaboration on Global Categorical Symmetries.


    \section{Overview of the Laurent polynomial method} \label{sec:Laurent}

    \subsection{Pauli operators}
    \label{sec:ringR_moduleP}

    We study quantum spin systems defined on the lattice $\mathbb{Z}^{d}$, with
    translation symmetry and local degrees of freedom at each site. The algebra $\mathcal{A}
    _{\lambda}$ of operators associated with any single point $\lambda \in \mathbb{Z}
    ^{d}$ is taken to be that of $m$ copies of the $n$-state clock and shift system,
    i.e.\ it is generated by operators $Z_{\lambda, j}, X_{\lambda,j}$,
    $j = 1, \dots, m$, satisfying the algebraic relations
    \begin{equation}
        Z_{\lambda,j}^{n} = X_{\lambda,j}^{n} = 1, \quad Z_{\lambda,j}X_{\lambda,j}
        = e^{\frac{2\pi i}{n}}X_{\lambda,j}Z_{\lambda,j}, \quad Z_{\lambda,j}X_{\lambda,k}
        = X_{\lambda,k}Z_{\lambda,j}\text{ for }j \neq k.
        \label{eq:ZX_commutation_rules}
    \end{equation}
    The operators $Z_{\lambda,j}$ commute among themselves, as do the operators
    $X_{\lambda,j}$. Commutation between $Z$ and $X$ operators is governed by the
    relations above. The algebra $\mathcal{A}_{\lambda}$ is the algebra of all
    linear operators on a complex vector space of dimension $n^{m}$. The~collection
    of operators
    \begin{equation}
        \prod_{j=1}^{m} X_{\lambda,j}^{a_{j}}Z_{\lambda,j}^{b_j},
    \end{equation}
    indexed by $2m$-tuples $(a_{j},b_{j})_{j=1}^{m} \in \mathbb{Z}_{n}^{2m}$ of integers
    modulo $n$, is a basis of $\mathcal{A}_{\lambda}$.

    The algebra $\mathcal{A}_{S}$ of operators supported in a finite subset $S$
    of $\mathbb{Z}^{d}$ is defined as the tensor product
    \begin{equation}
        \mathcal{A}_{S} = \bigotimes_{\lambda \in S}\mathcal{A}_{\lambda}.
    \end{equation}
    When $S \subset S'$, the algebra $\mathcal{A}_{S}$ embeds naturally into $\mathcal{A}
    _{S'}\cong \mathcal{A}_{S} \otimes \mathcal{A}_{S' \setminus S}$ via tensoring
    with $1 \in \mathcal{A}_{S' \setminus S}$. This identification allows one to
    define the algebra $\mathcal{A}_{S'}$ of local operators supported in an
    infinite subset $S' \subset \mathbb{Z}^{d}$ as the union of algebras
    $\mathcal{A}_{S}$ with finite $S$ contained in $S'$. We denote by $\mathcal{A}
    := \mathcal{A}_{\mathbb{Z}^d}$ the algebra of all local operators.

    We define the \textbf{Pauli group} $\mathcal{P}$ as the multiplicative group
    consisting of all operators
    \begin{equation}
        z \prod_{\lambda \in \mathbb{Z}^d}\prod_{j=1}^{m} X_{\lambda,j}^{a_{\lambda,j}}
        Z_{\lambda,j}^{b_{\lambda,j}}, \label{eq:Pauli_operator}
    \end{equation}
    where $z \in \mathbb{C}$ with $|z|=1$, $a_{\lambda,j}, b_{\lambda,j}\in \mathbb{Z}
    _{n}$, and only finitely many $a_{\lambda,j}, b_{\lambda,j}$ are nonzero. In
    other words, $\mathcal{P}$ consists of all finite products of the clock and
    shift operators multiplied by phase factors. We will call such operators \textbf{Pauli
    operators}.

    The \textbf{projective Pauli group} $P$ is defined as the quotient of $\mathcal{P}$
    by its center, consisting of all complex phase factors. Algebraically, $P$
    is isomorphic to the additive group of finitely supported functions
    $\mathbb{Z}^{d} \to \mathbb{Z}_{n}^{2m}$,
    \begin{equation}
        P \cong \bigoplus_{\lambda \in \mathbb{Z}^d}\mathbb{Z}_{n}^{2m}.
    \end{equation}
    We will denote the group operation in $P$ additively, even though it
    originates from operator multiplication modulo phase.

    The group $\mathcal{P}$ is a central extension of $P$ by the group of all
    complex scalars. Commutation relations of elements in $\mathcal{P}$ are described
    by a bilinear and alternating form $\omega : P \times P \to \mathbb{Z}_{n}$.
    For $p,p' \in P$ let $T_{p}, T_{p'}\in \mathcal{P}$ be representatives. Then:
    \begin{equation}
        T_{p} T_{p'}= T_{p'}T_{p} e^{\frac{2 \pi i}{n} \omega(p,p')}. \label{eq:curlyP_commutation}
    \end{equation}
    Explicitly, if $T_{p},T_{p}'$ is as in \eqref{eq:Pauli_operator}, with exponents
    $(a_{\lambda,j},b_{\lambda_j})$ and $(a_{\lambda,j}',b_{\lambda_j}')$,
    respectively, all finitely supported, then:
    \begin{equation}
        \omega(p,p') = \sum_{\lambda \in \mathbb{Z^d}}\sum_{j=1}^{m} (b_{\lambda,j}
        a_{\lambda,j}' - a_{\lambda,j}b_{\lambda,j}').
    \end{equation}
    One can reconstruct the group $\mathcal{P}$ from the data $(P,\omega)$, see
    Appendix \ref{app:central}.

    The lattice $\mathbb{Z}^{d}$ acts on $P$ by translations of the lattice
    sites, and this action endows $P$ with the structure of a module over the group
    ring of $\mathbb{Z}^{d}$ over $\mathbb{Z}_{n}$. This group ring is
    isomorphic to the ring of Laurent polynomials in $d$ variables,
    \begin{equation}
        R := \mathbb{Z}_{n}[x_{1}^{\pm 1}, \dots, x_{d}^{\pm 1}],
    \end{equation}
    where each variable $x_{i}$ corresponds to translation by one unit in the $i$-th
    lattice direction. We will use the multi-index notation for monomials in $R$:
    \begin{equation}
        x^{\lambda}:= \prod_{i=1}^{d} x_{i}^{\lambda_i}\qquad \text{for }\lambda
        = (\lambda_{1},\dots,\lambda_{d}) \in \mathbb{Z}^{d}.
    \end{equation}
    We will also denote by overline the $\mathbb{Z}_{n}$-linear involution on $R$
    which takes $x^{\lambda}$ to $x^{-\lambda}$. For example:
    \begin{equation}
        \overline{x_1 + 2x_2^{-1} x_3}= x_{1}^{-1}+2 x_{2} x_{3}^{-1}.
    \end{equation}

    By identifying finitely supported functions
    $\mathbb{Z}^{d} \to \mathbb{Z}_{n}$ with polynomials:
    \begin{equation}
        \bigoplus_{\lambda \in \mathbb{Z^d}}\mathbb{Z}_{n} \ni (c_{\lambda})_{\lambda
        \in \mathbb{Z^d}}\mapsto \sum_{\lambda \in \mathbb{Z^d}}c_{\lambda} x^{\lambda}
        \in R,
    \end{equation}
    we obtain an isomorphism
    \begin{equation}
        P \cong R^{2m}.
    \end{equation}
    Under this identification, each element $p \in P$ corresponds to a $2m$-tuple
    of Laurent polynomials
    \begin{equation}
        p = (p_{1}, p_{2}, \dots, p_{2m}) \in R^{2m}.
    \end{equation}
    For example, the tuple $p=(1, 0, \dots , 0 , x^{\lambda} + x^{\mu})$ is represented
    by the operator:
    \begin{equation}
        T_{p}= X_{0,1}Z_{\lambda,m}Z_{\mu,m}.
    \end{equation}
    The action of $R$ on $P$ is then given by coordinate-wise polynomial
    multiplication of Laurent polynomials:
    \begin{equation}
        f \cdot p = (f p_{1}, f p_{2}, \dots, f p_{2m}), \qquad f \in R.
    \end{equation}
    In particular, translation by a lattice vector $\lambda$ corresponds to
    multiplication by the monomial $x^{\lambda}$.

    The bilinear form $\omega: P \times P \to \mathbb{Z}_{n}$, which encodes the
    commutation relations of operators in $\mathcal{P}$, can be lifted to a
    polynomial-valued pairing $\Omega : P \times P \to R$. This enriched form captures
    not only the original commutation phase but also its behavior under translations.
    It is characterized by the following properties:
    \begin{itemize}
        \item $\Omega$ is $R$-sesquilinear, i.e. $\mathbb{Z}_{n}$-bilinear and for
            $p,p' \in P$ and $r \in R$ we have
            \begin{equation}
                \Omega(p,rp') = \Omega(\overline r p , p') = r \Omega(p,p'). \label{eq:Omega_sesqulinear}
            \end{equation}

        \item The constant term of $\Omega(p,p')$ is $\omega(p,p')$.
    \end{itemize}
    As a consequence, $\Omega$ is anti-hermitian: $\Omega(p,p') = - \overline{\Omega(p',p)}$
    for every $p,p' \in P$.

    The polynomial $\Omega(p,p')$ encodes not only the complex phase arising from
    the commutation of $T_{p}$ with $T_{p}'$, but also of all the translates. Indeed,
    the coefficient of $x^{\lambda}$ in $\Omega(p,p')$ is
    $\omega(p,x^{-\lambda}p')$, which also equals $\omega(x^{\lambda} p , p')$ by
    translation invariance or by \eqref{eq:Omega_sesqulinear}. Thus, the
    algebraic properties of $\Omega$ succinctly reflect the translation symmetry
    of the commutation relations.

    Under the identification $P \cong R^{2m}$, there exists a simple formula for
    $\Omega$ in terms of Laurent polynomials. Let
    $p=(p_{1},\dots,p_{2m}), p' = (p_{1}',\dots,p_{2m}') \in R^{2m}$. Then:
    \begin{equation}
        \Omega(p,p') =\sum_{j=1}^{m} (\overline{p_{m+j}}p'_{j} - \overline{p_j}p'
        _{m+j}).
    \end{equation}

    We remark that the setup discussed in this Section can be slightly generalized by allowing the exponent $n$ in \eqref{eq:ZX_commutation_rules} to depend on $j$; in other words, every lattice point could host several clock and shift systems with a different number of states. In this case free modules are replaced by a class of modules called quasi-free in \cite{ruba2024homological}.

    \subsection{Stabilizer codes}
    \label{sec:stabilizer}

    A large class of quantum spin systems of interest are \textbf{stabilizer codes}.
    They are defined by specifying a set of commuting Pauli operators, called
    \textbf{stabilizers}, whose common invariant vectors are the ground states. These states can be viewed as the ground states of a local Hamiltonian constructed
    from the stabilizers. While this viewpoint is important for the physical
    interpretation, the Hamiltonian operator typically plays a~secondary role in
    the mathematical analysis.

    Algebraically, a~stabilizer code is
    specified by a subgroup $L \subset P$ corresponding to a commuting family of operators; that is, $\omega(l, l') = 0$ for all $l, l' \in L$. For translation-symmetric stabilizer codes, $L$ is an $R$-submodule, and the commutativity condition implies that $\Omega(l,l')=0$ for all $l,l' \in L$.

    \begin{defn}
        The \textbf{orthogonal complement} of a subgroup $S \subset P$ is defined
        as
        \begin{equation}
            S^{\perp} := \{ p \in P \mid \omega(p, s) = 0 \text{ for all }s \in S
            \}.
        \end{equation}
        We say that $S$ is \textbf{isotropic} if $S \subset S^{\perp}$, i.e.\ if
        $\omega(s,s')=0$ for all $s,s' \in S$. 
        
        If $S = S^\perp$, we say that $S$ is \textbf{Lagrangian}.
    \end{defn}

    Given an isotropic subgroup $L \subset P$, the notion of its invariant
    vectors becomes meaningful only after lifting each $l \in L$ to an element of the Pauli group~$\mathcal{P}$. Such a~lift is unique only up to a phase
    factor. In what follows, we assume that a consistent choice of lifts, defining a group homomorphism $L \to \mathcal P$, has
    been fixed. With a slight abuse of notation we identify elements
    of $L$ with the corresponding operators in $\mathcal{P}$. See Appendix~\ref{app:stabilizer}
    for a more detailed discussion.

    A particularly important class of stabilizer codes is characterized by the condition that $L$ is Lagrangian.
    Physically, this implies that the stabilizer group fully
    constrains the system, so that the state fixed by all stabilizers, called the
    \textbf{stabilizer state} $\ket{0}$, is unique\footnote{In geometries other
    than $\mathbb{Z}^{d}$, such as finite volumes with periodic boundary conditions,
    the~ground state need not be unique. Instead, the system can possess
    multiple ground states that are locally indistinguishable.}; see Appendix \ref{app:stabilizer} for a precise statement. 
    
    Our primary
    interest is in stabilizer codes defined by Lagrangian submodules. They are sometimes
    referred to as \textbf{topological stabilizer codes}, as they satisfy the so-called
    local topological order condition \cite{haah2013commuting}, and frequently exhibit features
    believed to characterize phases of quantum matter that are distinct from
    product states. Notable examples include the toric code in two dimensions and
    more exotic models in three or more dimensions.

    An excited state, or an excitation, is a common eigenstate $\ket{s}$ of the
    stabilizers for which not all eigenvalues are equal to $1$. For a Lagrangian
    code, an excitation is labeled by its \textbf{syndrome map}: a group
    homomorphism $s_0: L \;\longrightarrow\; \mathbb{Z}_{n}$ such that
    \begin{equation}
        l\ket{s}=\exp \left( \frac{2 \pi i}{n}s_0(l) \right) \ket{s}.
    \end{equation}
    The subscript $0$ indicates that $s_0$ is the constant term of a map $s $, which takes $l$ to formal power series in $x_1^{\pm 1},\dots,x_d^{\pm 1}$ encoding values of $s_0$ on all translates of $l$. Particularly important are the \textbf{localized} excitations, for~which
    the state $\ket{s}$ differs from the ground state $\ket{0}$ only within a bounded region of $\mathbb{Z}^{d}$. This localization can be captured by requiring that $s$ is valued in polynomials, rather than general formal power series. Equivalently, for every $l \in L$ we have $s_0(x^{\lambda} l)=0$ for all but finitely many $\lambda \in \mathbb Z^d$. This discussion leads to the following definitions.
    \begin{defn} \label{def2}
\begin{itemize}
    \item The $R$-module of formal power series $\sum_{\lambda \in \mathbb Z^d} c_\lambda x^\lambda$, with no restrictions on the coefficients $c_\lambda \in \mathbb Z_n$, is denoted by $\widehat R$. 
    \item If $M$ is an $R$-module, we obtain a new $R$-module $\overline M$ whose underlying abelian group is still $M$, but whose ring action is twisted by the involution. That is, $\overline M$ consists of elements $m \in M$, with the action of $R$ defined by $(r,m) \mapsto \overline r m$. If $N$ is another $R$-module, we identify $\Hom_R(\overline M,N)$ with the module of $R$-antilinear maps $M \to N$.
    \item If $M$ is an $R$-module, we have a canonical isomorphism
    \begin{equation}
    \Hom_R(\overline M , \widehat R) \ni \alpha \mapsto \alpha_0 \in \Hom_{\mathbb Z_n}(M,\mathbb Z_n)
    \end{equation}
    given by extracting the constant term. The inverse map is given by
    \begin{equation}
        \alpha(m) = \sum_{\lambda \in \mathbb Z^d} \alpha_0(x^\lambda m) x^\lambda. \label{eq:localized}
    \end{equation}
    \item We say that $\alpha \in \Hom(\overline M,  \widehat R)$ and the corresponding $\alpha_0 \in \Hom_{\mathbb Z_n}(M,\mathbb Z_n)$ are \emph{localized} if $\alpha$ takes values in $R$, or, equivalently, if for every $m \in M$ we have \begin{equation}\label{eq: localized}\alpha_0 (x^\lambda m) =0 \quad \text{for all but finitely many} \quad \lambda \in \mathbb Z^d.\end{equation} We denote
    \begin{equation} \label{eq: star-functional corr}
        M^* := \Hom_R(\overline M,R).
    \end{equation}
\end{itemize}
\end{defn}

    \begin{exmp}
        Let $\ket{0}$ be the stabilizer state satisfying $l\ket{0}= \ket{0}$ for
        all $l\in L$. If~$T_{p} \in \mathcal{P}$ represents $p\in P$, the state $T
        _{p} \ket{0}$ is an excitation whose syndrome map is
        \begin{equation}
            L \ni l \mapsto \omega(l,p) \in \mathbb Z_n,
            \label{eq:locally_created_syndrome_map}
        \end{equation}
        with the polynomial-valued counterpart $\Omega(l,p)\in R$. Such excitations are localized in the sense of \eqref{eq: localized}.
        Two Pauli operators $p,q\in P$ give rise to the same syndrome map if and only if
        $p-q\in L$. Hence the excitations creatable by local Pauli operators are
        classified by the quotient $P/L$.
    \end{exmp}

    \begin{defn}
    The localized excitations are classified by the quotient
    \begin{equation}
        Q_L:=\frac{\text{localized excitation}}{\text{locally creatable excitation}}
        =\frac{L^{*}}{P/L}.
        \label{eq:Q_def}
    \end{equation}
    In other words, $Q_L$ is the cokernel of the map
    \begin{equation}
        P \ni p \mapsto \left. \omega( \cdot , p ) \right|_{L} \in L^{*}.
    \end{equation}
    \end{defn}

    In~\cite{ruba2024homological}, it was shown that the module $Q_L$ can be expressed
    using homological algebra as $\Ext_{R}^{1}(\overline{P/L}, R)$. This
    characterization is pivotal for three reasons:
    \begin{enumerate}
        \item It allows to prove general results about the structure of stabilizer
            codes;

        \item It facilitates computations, since using homological methods is
            often more efficient than working from the definition;

        \item It inspires generalization of $Q_L$ to higher invariants of stabilizer
            codes.
    \end{enumerate}
    We elaborate on point 3 above. In \cite{ruba2024homological} it was proposed to define higher invariants of stabilizer codes by considering derived functors in higher degrees. Specifically, 
    \begin{equation}
     Q_L^i:= \mathrm{Ext}^{i+1}(\overline{P/L},R).   
     \label{eq:Qi_def}
    \end{equation}
    This definition, which reproduces \eqref{eq:Q_def} for $i=0$, is rather abstract. As shown in~\cite{ruba2024homological}, the groups $Q_L^i$ can be interpreted in terms of extended excitations, such as line or membrane-like defects, as well as in terms of higher symmetries. In the present work, we will not consider the case $i>0$.

    In the study of Pauli stabilizer codes, \textbf{coarse-graining} is the process of grouping together degrees of freedom associated to nearby lattice points in order to focus on the large-scale structure while discarding microscopic details. Algebraically, this is implemented by restricting the group of translations from $\mathbb Z^d$ to some finite index subgroup $\Lambda \subset \mathbb Z^d$. More precisely, coarse-graining is the operation of restriction of scalars: $R$-modules $M$ are studied as modules over the subring $R_\Lambda \subset R$ spanned by the monomials $x^{\lambda}$ with $\lambda \in \Lambda$. We highlight that $\Lambda \cong \mathbb Z^d$, so~$R_\Lambda$ is non-canonically isomorphic to $R$. Moreover, $R$ is a finitely generated free module over $R_\Lambda$, so freeness and finite generation are properties preserved by coarse-graining.

    Importantly, the locality of functionals is coarse-graining invariant: if~$M$ is an $R$-module and $\alpha_0 : M \to \mathbb Z_n$ is a group homomorphism, then the property of $\alpha_0$ being localized (cf.\ \eqref{eq: localized}) does not depend on whether $M$ is considered as an $R$-module or an $R_\Lambda$-module. As a consequence \cite[Proposition 34]{ruba2024homological}, the construction $Q^i$ in \eqref{eq:Qi_def} commutes with coarse-graining: if $L \subset P$ arises from $L' \subset P'$ by coarse-graining, then $Q^i_L$ is obtained from $Q^i_{L'}$ by coarse-graining.
    
\subsection{Laurent polynomial rings} \label{sec:23}

In the final part of this section, we review some standard facts about Laurent polynomial rings $
R := \mathbb{Z}_n[x_1^{\pm1},\dots,x_d^{\pm1}]$ that will be used later. More details can be found in the Appendix of~\cite{ruba2024homological}. 

Let $ n = \prod_{i=1}^k p_i^{r_i} $ be the prime factorization of $n$.
Then the Chinese Remainder Theorem gives a canonical isomorphism
\begin{equation}
R \cong \prod_{i=1}^k \mathbb{Z}_{p_i^{r_i}}[x_1^{\pm1},\dots,x_d^{\pm1}],
\end{equation}
and a corresponding decomposition of $R$-modules. All module-theoretic constructions decompose accordingly and can be analyzed prime-by-prime. Thus, it suffices to study the case where \(n = p^r\) is a prime power. 

If $n=p^r$, every element of $R$ is either \textbf{regular} (i.e.\ not a zero divisor) or nilpotent. The prime ideals of $R$ are in one-to-one correspondence with those in $\mathbb F_p[x_1^{\pm 1},\dots,x_d^{\pm 1}]$. In particular, the Krull dimension of $R$ is $d$.


\begin{defn}
An \( R \)-module \( E \) is called an injective cogenerator (in the category of $R$-modules) if it is an injective $R$-module and for every nonzero \( R \)-module \( M \), there exists a nonzero homomorphism \( M \to E \).
\end{defn}

\begin{prop}
    The module $\widehat R$ of formal power series is an injective cogenerator in the category of $R$-modules, and for every $R$-module $M$
    \begin{equation}
        \Hom_R(\overline M, \widehat R)\cong \Hom_{\ZZ_n}(M, \ZZ_n).
            \label{eq:Rhat_duality}
    \end{equation}
\end{prop}
\begin{proof}
    The isomorphism \eqref{eq:Rhat_duality} is given explicitly in Definition \ref{def2}. The rest of the claim follows because $\mathbb Z_n$ is an injective cogenerator in the category of $\mathbb Z_n$-modules.
    \end{proof}
    

When $n$ is not a prime number, $R$ has zero divisors, and the theory of $R$-modules is more complicated than in the case of prime $n$. A good property of $R$ which facilitates arguments based on homological algebra is that it is a Gorenstein ring (see e.g.\ \cite{Bruns_Herzog}). We record some consequences below. 

\begin{prop} \label{prop:Ext_bounds}
    Let $M$ be an $R$-module, then 
    \begin{equation}
        \Ext^i_R(M, R)=0 \text{ for } i>d. 
    \end{equation}
    If $M$ is finitely generated, Krull dimension of its Ext modules obeys the bound
    \begin{equation}
        \dim (\Ext^i(M,R)) \leq d -i. 
    \end{equation}
    If in addition $M$ is torsion-free, then 
    \begin{equation}
      \dim (\Ext^i(M,R)) \leq d -i-1, \qquad  \Ext^d_R(M, R)=0. 
    \end{equation}
\end{prop}
%

\begin{prop}{(Propositions 5.1 and 5.3 in \cite{cartan1999homological})}
Let $A$, $B$, and $C$ be $R$-modules, with $A$ finitely generated and $C$ injective. Then for all $i \geq 0$, there are canonical isomorphisms
\begin{equation}
\sigma_i: \Tor_i^R(A, \Hom_R(B, C)) \cong \Hom_R(\Ext^i_R(A, B), C),
\end{equation}
\begin{equation}
\rho^i: \Ext^i_R(A, \Hom_R(B, C)) \cong \Hom_R(\Tor_i^R(A, B), C).
\end{equation}
\end{prop}

\begin{cor} \label{cor:tor_ext}
    Let $M$ be a finitely generated $R$-modules, for all $i \geq 0$, there is a~canonical isomorphism
    \begin{equation}
        \Tor_i^R(M, \widehat R) \cong \Hom_R(\Ext^i_R(M, R), \widehat R).
    \end{equation}
    In particular, \begin{equation} \label{eq: Ext_Tor_vanish1}
        \Tor_i^R(M, \widehat R)=0 \text{ if and only if } \Ext^i_R(M, R)=0.
    \end{equation} 
\end{cor}
\begin{proof}
    Let $A=M, B= R$ and $C= \widehat R$ in the proposition above. Then \eqref{eq: Ext_Tor_vanish1} is deduced from the fact that $\widehat R$ is a cogenerator. 
\end{proof}

By Corollary \ref{cor:tor_ext}, $\mathbb Z_n$-duals of charge modules of a stabilizer code described by a~Lagrangian $L \subset P$ are given by Tor modules:
\begin{equation}
    \mathrm{Tor}_{i+1}^R (P/L,\widehat R) \cong \Hom_R (Q^i_L, \widehat R) \cong \Hom_{\mathbb Z_n}(Q^i_L, \mathbb Z_n).
\end{equation}
In the case $i=0$, this has the simple interpretation that every nontrivial localized excitation can be detected by a non-local operator which is a product of infinitely many stabilizers which locally cancel out, yielding an operator which is well-defined on local excitations and acts trivially on locally creatable excitations. Indeed, the long exact sequence of $\mathrm{Tor}_\bullet^R(\cdot,\widehat R)$ modules induced by the short exact sequence $0 \to  L \to  P \to P/L \to 0 $ exhibits $\mathrm{Tor}_1^R(P/L,\widehat R)$ as the kernel
\begin{equation}
 \mathrm{Tor}_1^R(P/L,\widehat R) \cong \ker(  L \otimes_R \widehat R \to  P \otimes_R \widehat R)
 \label{eq:Tor_as_Ker}
\end{equation}
of the map induced by the inclusion $L \subset P$. The modules $L \otimes_R \widehat R$ and $P \otimes_R \widehat R$ describe infinite products of stabilizers and of general Pauli operators, respectively. Naively, one might expect that elements in the kernel \eqref{eq:Tor_as_Ker} must correspond to trivial operators, but this is wrong. As an example, consider the product of all plaquette operators in the toric code. One can define the action of such product on any state which is a local excitation of the ground state, and it acts on the $m$ excitation as multiplication by $-1$.

    \section{Systems without on-site product structure} \label{sec:qs}


    \subsection{Quasi-symplectic modules} \label{sec:quasi-symp}

    We now introduce the notion of quasi-symplectic modules $(P, \Omega)$, which
    generalize the structure discussed in Section~\ref{sec:ringR_moduleP}. The
    specific case discussed there will be referred to as the \textbf{standard
    symplectic module}.

    \begin{defn}
        A \textbf{quasi-symplectic module} is a finitely generated $R$-module $P$
        equipped with a sesquilinear form $\Omega : P \times P \to R$ satisfying
        \begin{itemize}
            \item For every $p \in P$, the constant term of $\Omega(p,p)$
                vanishes,

            \item If $\Omega(\cdot, p) =0$, then $p=0$.
        \end{itemize}
        It follows from these properties that $\Omega$ is anti-hermitian, i.e.\ $\Omega
        (p,p')=-\overline{\Omega(p',p)}$.
    \end{defn}

    Given a quasi-symplectic module $(P,\Omega)$, one obtains an alternating
    bilinear form $\omega : P \times P \to \mathbb{Z}_{n}$ by taking the
    constant term of $\Omega$. This form determines a~unique (up~to isomorphism)
    central extension $\mathcal{P}$ of $P$ by the group of complex phase factors,
    in which the commutation relations are governed by $\omega$ as described in
    \eqref{eq:curlyP_commutation} (see~Appendix~\ref{app:central}).

    Importantly, the fact that $\omega$ arises as the constant term of an $R$-valued
    form $\Omega$ reflects the \emph{locality} of the commutation relations: for
    any $p, p' \in P$, the quantity $\omega(p, x^{\lambda}p')$ vanishes for all but
    finitely many $\lambda \in \mathbb{Z}^{d}$. That is, each operator commutes with
    all but finitely many lattice translates of another. This locality property
    serves as a substitute for the on-site nature of degrees of freedom encoded
    in standard symplectic modules.

    The notions of orthogonal complements and of isotropic and Lagrangian
    submodules are defined for quasi-symplectic modules in the same way as for standard
    symplectic modules in Section~\ref{sec:stabilizer}. One can also construct stabilizer states corresponding to Lagrangian submodules, see Appendix \ref{app:stabilizer}.

    \begin{exmp}
        If $(P,\Omega)$ is a quasi-symplectic module and $M \subset P$ is a submodule,
        then $\left. \Omega \right|_{M}$ descends to the quotient
        $M / (M \cap M^{\perp})$, endowing this quotient with the structure of a
        quasi-symplectic module.

        This construction can be used, for example, to study Pauli operators that are invariant
        under some symmetry group acting on $\mathcal{P}$ by automorphisms. The algebra of such invariant local operators does not
        necessarily have the structure of the tensor product of on-site operator
        algebras. Nevertheless, it can still be described using a~quasi-symplectic
        module. 
    \end{exmp}

    Later, we will encounter a natural construction that gives rise to quasi-symplectic
    modules: the algebra of boundary degrees of freedom of a stabilizer code
    defined on a half-space geometry. First, however, let us develop some general
    theory that is applicable regardless of the specific context in which a quasi-symplectic
    module arises. As we will see, it produces certain invariants which allow to distinguish truly nontrivial quasi-symplectic modules from standard symplectic modules. 

    \begin{prop} \label{prop: quasi-sym}
    \cite{ruba2024homological} Let $(P,\Omega)$ be a quasi-symplectic module
        and let $M \subset P$ be a~submodule.
        \begin{enumerate}
            \item $P$ is a torsion-free module.

            \item The cokernel of the map
                \begin{equation}
                    P \ni p \mapsto \left. \Omega(\cdot , p) \right|_{M} \in M^{*}
                    \label{eq:quasi_symp_cokernel}
                \end{equation}
                is a torsion module.

            \item The double orthogonal complement $M^{\perp \perp}$ coincides with
                the saturation of $M$ in $P$; that is, it is the set of all $p \in P$
                such that there exists a regular element $r \in R$ such that $rp
                \in M$.
                \item We have $M^{\perp \perp \perp} = M^{\perp}$. In particular, $(\perp \perp)^2 = \perp \perp$.
        \end{enumerate}
    \end{prop}

    \begin{defn}
        For a quasi-symplectic module $P$, we let $E_P$ be the module
        \begin{equation}
            E_P :=P^{*} / P,
        \end{equation}
        or, equivalently, the cokernel of \eqref{eq:quasi_symp_cokernel} for $M=P$. 

       When $E_P$ is the trivial module, that is when $\Omega$ induces an isomorphism $P \cong P^*$, we
    call $P$ a \textbf{symplectic module}. As suggested by the adapted terminology,
    standard symplectic modules are examples of symplectic modules.
    \end{defn}

    The operation $\perp \perp$ has properties analogous to topological closure. In the remainder of this Section, we describe certain natural topologies in which $\perp \perp$ serves as the closure operator when applied to submodules. This construction is not essential to follow the rest of the text.

    First consider the rank one free module $R$. As an abelian group, \begin{equation}
        R \cong \bigoplus_{\lambda \in \mathbb Z^d} \mathbb Z_n \subset \prod_{\lambda \in \mathbb Z^d} \mathbb Z_n.
    \end{equation} 
    We equip $R$ with the subspace topology induced by the product topology on $\prod_{\lambda \in \mathbb Z^d} \mathbb Z_n$. This topology is metrizable, and a sequence of polynomials $r_j \in R$ converges to $r \in R$ if and only if for every $\lambda \in \mathbb Z^d$ the coefficient of $r_j-r$ at $x^\lambda$ vanishes for sufficiently large $j$ (dependent on $\lambda$); we think of $r_j-r$ as escaping to infinity in~$\mathbb Z^d$. The topology we described makes $R$ a topological abelian group with addition, but multiplication is only separately continuous. 
    
    Finite rank free modules $R^t$ are given the product topology. If $M$ is any finitely generated module, we choose a finite generating set $e_1,\dots,e_t$ in $M$ and equip $M$ with the quotient topology induced by the epimorphism
    \begin{equation}
        R^t \ni (f_1,\dots,f_t) \mapsto \sum_{i=1}^t f_i e_i \in M.
    \end{equation}
    If $m_j$ is a sequence in $M$, then $m_j$ converges to an element $m \in M$ if and only if for every finite subset $\Gamma \subset \mathbb Z^d$ and for sufficiently large $j$ (dependent on $\Gamma$)
    \begin{equation}
        m-m_j \in \sum_{i=1}^t \sum_{\lambda \in \mathbb Z^d \setminus \Gamma} \mathbb Z_n x^\lambda e_i. 
    \end{equation}

\begin{prop} \label{prop:topologies}
    Let $M$, $N$ be finitely generated modules carrying quotient topologies induced by epimorphisms $\sigma_M : R^t \to M$ and $\sigma_N : R^s \to N$, corresponding to some choices of generators. 
    \begin{enumerate}
        \item Every module homomorphism $M \to N$ is continuous.
        \item The topology on $M$ does not depend on the choice of generators.
        \item The closure of $\{ 0 \}$ in $M$ is the torsion submodule of $M$. In particular, $M$ is Hausdorff if and only if it is torsion-free.
    \end{enumerate}
\end{prop}
\begin{proof}
    1. Let $\phi : M \to N$ be a module homomorphism. We have to show that $\phi \circ \sigma_M$ is continuous. By freeness, there exists a homomorphism $\phi' : R^t \to R^s$ such that $\sigma_N \circ \phi' = \phi \circ \sigma_M$. It is easy to check that homomorphisms of free modules are continuous. Thus $\sigma_N \circ \phi' = \phi \circ \sigma_M$ is continuous, and hence $\phi$ is continuous by the universal property of the quotient map $\sigma_N$.

    2. By 1., the identity endomorphism of $M$ is a homeomorphism between topologies defined by any two epimorphisms $R^t \to M$ and $R^s \to M$. 

    3. First we prove $\supset$. We can assume that $n$ is a prime power. Let $m \in M$ be annihilated by a regular element $f \in R$, and let
\begin{equation}
    S = \{ \lambda \in \mathbb Z^d \, | \, \text{the coefficient of } f \text{ at } x^\lambda \text{ is invertible} \}.
\end{equation}
$S$ is a finite set, and it is nonempty because $f$ is regular. If $S$ has only one element, then $f$ is invertible and hence $m=0$. Thus we may assume that $S$ has at least two elements. Let $S_{\mathrm{ext}}$ be the set of extreme points of the convex hull of $S$ in $\mathbb R^d$. In other words, $S_{\mathrm{ext}}$ is the set of $s \in S$ which are not in the convex hull of $S \setminus \{ s \}$. We have $S_{\mathrm{ext}} \neq \emptyset$. We replace $f$ by $x^\lambda f$ for some $\lambda \in \mathbb Z^d$ to ensure that $0\in S_{\mathrm{ext}}$. We~can also assume that the constant term of $f$ equals $1$.

It is well known that every vertex of a convex polytope in $\mathbb R^d$ is the unique minimizer over the polytope of some linear functional. Therefore, there exists $\delta \in \mathbb R^d$ such that for every $\lambda \in S \setminus \{ 0 \}$ we have 
\begin{equation}
    \delta \cdot \lambda = \sum_{i=1}^d \delta_i \lambda_i > 0 . 
\end{equation}
We let $\delta_0 > 0$ be the minimum of $\delta \cdot \lambda$ over $S \setminus \{ 0 \}$. 

We express $f$ in the form $1-g-h$, where $-g$ contains the terms of $f$ proportional to $x^\lambda$ with $\lambda \in S \setminus \{ 0 \}$, and $-h$ contains the terms proportional to $x^\lambda$ with $\lambda \not \in S$. Then $h$ is nilpotent. We let $r$ be an integer such that $h^r =0$. We have
\begin{equation}
    m = (g+h) m, 
\end{equation}
and therefore for any integer $N > 0$:
\begin{equation}
    m = (g+h)^N m.
\end{equation}
We take $N \geq r $ and perform the binomial expansion:
\begin{equation}
    m = \left( \sum_{k=0}^r \binom{N}{k} g^{N-k} h^k  \right) m .
    \label{eq:m_equals_m_far_away}
\end{equation}
The polynomial in the parenthesis contains only terms proportional to $x^\lambda$ with
\begin{equation}
    \delta \cdot \lambda \geq N \delta_0  + c,
\end{equation}
where $c$ is a constant independent of $N$. Hence the right hand side of \eqref{eq:m_equals_m_far_away} converges to zero as $N \to \infty$. This completes the proof of $\supset$. 

 To establish $\subset$, note that if $m \in M$ is in the closure of $\{ 0 \}$, then $\varphi(m) =0$ for all $\varphi \in M^*$. Then $m$ is a torsion element, see e.g.~\cite[Proof of Lemma 8]{ruba2024homological}.
\end{proof}

We remark that if $M' \subset M$ is a submodule, the subspace of topology induced on $M'$ may be strictly coarser than our preferred topology induced by a choice of generators in $M'$. A simple example of this is $M = R$ and $M'=(x-1)$, the ideal generated by $x-1$. Then the sequence
\begin{equation}
    x^j - x^{-j} = (x-1) \sum_{k=-j}^{j-1} x^k
\end{equation}
converges to zero as $j \to \infty$ in the module $M$, but not in $M'$. 

\begin{prop} \label{prop:double_perp_closure}
    Let $(P,\Omega)$ be a quasi-symplectic module and let $M \subset P$ be a~submodule. Then $M^{\perp \perp}$ is the closure of $M$ in $P$.
\end{prop}
\begin{proof}
The quasi-symplectic form $\Omega$ is separately continuous, so the closure of $M$ is contained in $M^{\perp \perp}$. For the opposite inclusion, if $m \in M^{\perp \perp}$, then by 3. in Proposition \ref{prop: quasi-sym} there exists a regular element $f \in R$ such that $fm \in M$. That is, $fm$ represents zero in $P/M$. Proceeding as in the proof of 3. in Proposition \ref{prop:topologies}, we~construct a~sequence of elements $r_j \in R$ such that
\begin{equation}
    m = r_j m \text{ mod } M, \qquad r_j \to 0.
\end{equation}
Thus we can find a sequence $m_j$ in $M$ such that $m=r_j m + m_j$. Then $m_j \to m$.
\end{proof}

 We remark that in standard symplectic modules, $\perp \perp$ coincides with the closure when applied to arbitrary subgroups. This follows easily from the Pontryagin duality between $R^{2q}$ regarded as a discrete group and the compact group $\widehat R^{2q} \cong \prod_{\lambda \in \mathbb Z^d} \mathbb Z_n^{2q}$. The result in Proposition \ref{prop:double_perp_closure} is more general in that it allows for arbitrary quasi-symplectic modules, but it applies only to submodules. This is because it relies on the algebraic description of $\perp \perp$ given in 3. of Proposition \ref{prop: quasi-sym}.

      \subsection{Witt group for quasi-symplectic modules}
    \label{sec:weak_witt}

This subsection develops a classification theory for quasi-symplectic modules over $R = \mathbb{Z}_n[x_1^{\pm 1}, \dots, x_d^{\pm 1}]$ via the \textbf{Witt group}. Classically, the Witt group (or algebraic $L$-group~\cite{ranicki1973algebraicI}) classifies (skew-)hermitian or quadratic forms up to stable equivalence, identifying as trivial those that are \textbf{metabolic}, i.e. those admitting a~Lagrangian submodule that is a~direct summand. The weak Witt group~\cite{scharlau2012quadratic} relaxes this, not requiring Lagrangian submodules of metabolic forms to be direct summands. Adopting this convention, we refer to our construction simply as the Witt group, though some authors may call it the weak Witt group. Another difference between our treatment and a part of the literature is that we do not restrict attention to projective or free modules. 

A major advantage of the Witt group is that it addresses a fundamental question--namely, which quasi-symplectic modules admit a Lagrangian submodule--without requiring a classification up to isomorphism, which is substantially more difficult. As we demonstrate below, calculations in the Witt group can be reduced to the case of prime $n$, whereas no such reduction appears to be available for classification up to isomorphism.

    \begin{defn}
    Given two quasi-symplectic modules $(P, \Omega)$ and $(P', \Omega')$ over $R$, their \textbf{orthogonal sum} is $(P \oplus P', \Omega \oplus \Omega')$, and the \textbf{opposite form} of $(P, \Omega)$ is $(P, -\Omega)$. Both of them are quasi-symplectic. Define the shorthand notation 
    \begin{equation}
     P\ominus P' := (P \oplus P', \Omega \oplus -\Omega').   
    \end{equation}
    \end{defn}

    \begin{defn}
    We say that $(P, \Omega)$ and $(P', \Omega')$ are \textbf{isomorphic} if there exists an $R$-module isomorphism $\phi: P \to P'$ such that $ \Omega' \circ (\phi \times \phi) = \Omega$. We say they are Witt equivalent, or~write $(P, \Omega)\sim (P', \Omega')$, if $P\ominus P'$ contains a Lagrangian submodule. In~particular, $(P, \Omega)$ contains a Lagrangian submodule if and only if $(P, \Omega)\sim 0$. In~this case $P$ is said to be metabolic.
    \end{defn}
    It is easy to see that isomorphism between quasi-symplectic modules defines an equivalence relation. To show that Witt equivalence is also an equivalence relation, two~technical lemmas are needed.
        \begin{lem} \label{lem: inverse_form}
\begin{enumerate}
    \item The orthogonal sum of any two metabolic modules is metabolic.
    \item $P\ominus P$ is metabolic for any quasi-symplectic module $(P, \Omega)$; in fact it admits a~Lagrangian direct summand.
\end{enumerate}
    \end{lem}
    \begin{proof}
    1. is obvious. To prove 2., let $L$ be the diagonal submodule of $P \oplus P$, defined by $L = \{ (p, p) \mid p \in P \}$. We claim that $L^\perp = L$. Indeed, for any $p, q \in P$, we have
    \begin{equation}
    (\Omega \oplus -\Omega)\big((p, p), (q, q)\big) = \Omega(p, q) - \Omega(p, q) = 0.
    \end{equation}    Now, let $(a, b) \in P \oplus P$ be in $L^\perp$. This means
    \begin{equation}
    \Omega(a, p) - \Omega(b, p) = 0 \quad \text{for all } p \in P,
    \end{equation}
    so $\Omega(a-b, p) = 0$ for all $p \in P$. By nondegeneracy, $a = b$, so $(a, b) \in L$. We showed that $L^\perp = L$. Pairs $(p,0) \in P \ominus P$ form a complementary submodule, so $L$ is a~direct summand.
    \end{proof}
   \begin{lem} \label{lem: cancellative} 
        Let $(P, \Omega)$ and $(P', \Omega')$ be quasi-symplectic modules such that $L\subset P$ and $K\subset P\oplus P'$ are Lagrangian respectively. Let
        \begin{equation}
            L'_0:=\{y\in P': \text{ there exists } x\in L \text{ with } (x,y)\in K\}, \qquad \qquad L' = L_0'^{\perp \perp}.
        \end{equation}
        Then $L'$ is Lagrangian in $P'$.
    \end{lem}
\begin{proof} 
    Given $y_1, y_2\in L_0'$, let $x_1, x_2\in L$ be such that $(x_1, y_1),(x_2, y_2)\in K$. Then 
    \begin{equation}
     \Omega'(y_1, y_2)=(\Omega \oplus \Omega')((x_1,y_1), (x_2, y_2))- \Omega(x_1, x_2)=0.   
    \end{equation}
    Hence, $L_0'\subset L_0'^\perp$, which implies 
    \begin{equation}
    L'=L_0'^{\perp \perp}\subset L_0'^{\perp \perp\perp}=L'^\perp.
    \end{equation}

Conversely, if $y \in L'^{\perp} = L_0'^{\perp}$, then by definition of $L'_0$,
\begin{equation}
  (0,y)\in (K\cap(L\oplus P'))^{\perp} = (K + (L \oplus 0))^{\perp \perp}.
\end{equation}
By 3. in Proposition \ref{prop: quasi-sym}, there exist a regular element $r \in R$ and $(a,x) \in K$, $l \in L$ such that
\begin{equation}
    (0,ry) = (a,x) + (l,0).
\end{equation}
Comparing components, we get that $a=-l \in L$, and hence $ry=x \in L_0'$. Thus $y \in L'$.
\end{proof}

    \begin{prop}
        The Witt equivalence is an equivalence relation.
    \end{prop}
        \begin{proof}
    Reflexivity holds because of Lemma~\ref{lem: inverse_form}. Symmetry is immediate. For transitivity, suppose $(P, \Omega) \sim (P', \Omega')$ via a Lagrangian $L \subset P \ominus P'$, and $(P', \Omega') \sim (P'', \Omega'')$ via a Lagrangian $L' \subset P' \ominus P''$. Then $L \oplus L'$ is Lagrangian in $P \ominus P'\oplus P' \ominus P''$. By Lemma~\ref{lem: inverse_form}, $(P' \oplus P', -\Omega' \oplus \Omega')$ is metabolic, and Lemma~\ref{lem: cancellative} implies that $P \ominus P''$ is metabolic. Thus, $(P, \Omega) \sim (P'', \Omega'')$.
    \end{proof}

        \begin{defn}
    Let $\mathcal{W}$ denote the set of isomorphism classes of quasi-symplectic modules over $R$. This set forms a commutative monoid under orthogonal sum, with the metabolic modules forming a submonoid.

    The Witt group $W(R)$ of quasi-symplectic modules is defined as the set of Witt equivalence classes $\mathcal{W}/\sim$. This name is justified below.
    \end{defn}

\begin{prop}
    The set $W(R)$ is a group. 
\end{prop}
\begin{proof}
    The quotient $\mathcal{W} /\sim$ carries an induced monoid structure because 
    \begin{equation*}
     (P, \Omega)\sim (P', \Omega') \text{ and }  (M, \Theta)\sim (M', \Theta') \implies (P\oplus M, \Omega\oplus \Theta)\sim(P'\oplus M', \Omega'\oplus \Theta').
    \end{equation*}
    Moreover, $W(R)$ is a group because every $(P, \Omega)$ has an inverse $(P, -\Omega)$.
\end{proof}

Recall that a quasi-symplectic module $(P, \Omega)$ is called symplectic if $E_P$ is trivial. Clearly, orthogonal sum of two symplectic modules is symplectic, and the opposite of a symplectic module is symplectic. Therefore, Witt equivalence classes of symplectic modules form a subgroup of $W(R)$.

\begin{defn}
 We let $W^s(R)$ be the subgroup of $W(R)$ consisting of Witt equivalence classes of symplectic modules.   
\end{defn}

   \begin{lem} \label{lem:reduction_Witt_equivalent}
Let $(P, \Omega)$ be a quasi-symplectic module and $N \subset P$ a submodule satisfying $N^\perp \subset N$. Then $N$ embeds as a Lagrangian submodule in $P \ominus N/ N^{\perp}$. In~particular, $P$ and $N / N^{\perp}$ represent the same elements in the Witt group $W(R)$.

If $P$ is symplectic, $N = N^{\perp \perp}$, and every homomorphism $N \to R$ admits an extension to $P$, then $N/N^{\perp}$ is symplectic. Hence $P$ and $N/N^{\perp}$ are equal in $W^s(R)$.
\end{lem}
\begin{proof}
Define the map $i: N \to P \oplus N / N^{\perp }$ by $x \mapsto (x, [x])$. Verification that the image of $i$ is Lagrangian is analogous to the proof of Lemma \ref{lem: inverse_form}.

We identify $(N/N^{\perp})^*$ with elements of $N^*$ vanishing on $N^{\perp}$. By the assumptions, we have surjections $P \to P^* \to N^*$, and an element $p \in P$ is mapped to $(N/N^{\perp})^*$ if and only if $p \in N^{\perp \perp} = N$. Hence the symplectic form defines an epimorphism $N \to (N/N^{\perp})^*$ with kernel $N^\perp$, i.e.\ an isomorphism $N/N^{\perp} \to (N/N^{\perp})^*$.
\end{proof}

The Chinese remainder theorem implies that
\begin{equation}
W(\mathbb{Z}_n[x_1^{\pm 1}, \dots, x_d^{\pm 1}])\cong \bigoplus_k W(\mathbb{Z}_{p_i^{r_i}}[x_1^{\pm 1}, \dots, x_d^{\pm 1}]),
\end{equation}
where $n = \prod_{i=1}^k p_i^{r_i}$ is the prime factorization of $n$: $p_i$ are distinct primes and $r_i \geq 1$. This reduces the theory to the case of $n$ being a prime power, $n=p^r$ for some prime $p$ and $r \geq 1$. The next result allows to reduce to the case $r=1$.

\begin{thm} \label{thm:reduction_to_prime}
Any quasi-symplectic module $(P, \Omega)$ is equivalent in the Witt group to a quasi-symplectic module $(P', \Omega')$ with $pP'=0$. 

\end{thm}

\begin{proof}
    Let $s$ be the smallest positive integer such that $p^sP=0$. If $s=1$, the~theorem's condition is met. If $s>1$, set $t = \lceil s/2 \rceil$; note that $0 < t < s$. The~submodule $N = p^t P$ is isotropic. Consequently, by Lemma \ref{lem:reduction_Witt_equivalent}, $P$ is Witt equivalent to $P_1 = N^\perp / N^{\perp\perp}$. Furthermore, this module $P_1$ is annihilated by $p^t$. Since $t < s$, this step reduces the exponent of $p$ annihilating the module. Repeating this procedure eventually yields a quasi-symplectic module $P'$ Witt equivalent to $P$ and satisfying $pP'=0$ (possibly $P'=0$).
\end{proof}

\begin{cor} \label{cor:W_reduction_iso}
    There is an isomorphism 
    \begin{equation}
    W(\mathbb{F}_p[x_1^{\pm 1}, \dots, x_d^{\pm 1}])\rightarrow W(\ZZ_{p^r}[x_1^{\pm 1}, \dots, x_d^{\pm 1}]), \qquad (P, \Omega)\mapsto(P, [p^{r-1}]\Omega).
    \label{eq:W_iso_prime}
    \end{equation}
    On the right-hand side, $P$ is regarded as a~$\ZZ_{p^r}[x_1^{\pm 1}, \dots, x_d^{\pm 1}]$-module via the quotient map $\ZZ_{p^r}[x_1^{\pm 1}, \dots, x_d^{\pm 1}] \to \mathbb{F}_p[x_1^{\pm 1}, \dots, x_d^{\pm 1}]$. Its quasi-symplectic form $\Omega$ is turned into one valued in $\ZZ_{p^r}[x_1^{\pm 1}, \dots, x_d^{\pm 1}]$ using the inclusion
    \begin{equation}
    [p^{r-1}]:\mathbb{F}_p[x_1^{\pm 1}, \dots, x_d^{\pm 1}] \to \ZZ_{p^r}[x_1^{\pm 1}, \dots, x_d^{\pm 1}], \qquad f \mapsto p^{r-1} f.
    \end{equation}
\end{cor}
\begin{proof}
    A quasi-symplectic $\mathbb{F}_p[x_1^{\pm 1}, \dots, x_d^{\pm 1}]$-module $(P, \Omega)$ is metabolic if and only if $(P, [p^{r-1}]\Omega)$ is metabolic. Thus the map is injective. As for surjectivity, given a~quasi-symplectic $\ZZ_{p^r}[x_1^{\pm 1}, \dots, x_d^{\pm 1}]$-module $(P', \Omega')$, one can assume that $pP'=0$ by \ref{thm:reduction_to_prime}. Then for any $a, b\in P'$ one has
    \begin{equation}
     p\Omega'(a, b)=0, \text{  hence  }   \Omega'(a, b)\in (p^{r-1})\subset R.
    \end{equation}
    Hence $(P,\Omega)$ is in the image of \eqref{eq:W_iso_prime}.
\end{proof}

This reduction to modules over $\mathbb{F}_p[x_1^{\pm 1}, \dots, x_d^{\pm 1}]$ is particularly advantageous in one dimension ($d=1$). Indeed, the ring $\mathbb{F}_p[x^{\pm 1}]$ is a principal ideal domain, over which all finitely generated torsion-free modules are free. The classification of quasi-symplectic modules is a problem in matrix algebra over $\mathbb{F}_p[x^{\pm 1}]$, and has been analyzed in \cite{haah2025topological}. We record an important result below, before dedicating the subsequent subsection to the one-dimensional case.

\begin{lem} \label{lem:surjective_res}
    Let $d=1$. Suppose that $(P, \Omega)$ is a quasi-symplectic module, and $M \subset P$ is a submodule satisfying $M=M^{\perp \perp}$. The restriction homomorphism $P^* \to M^*$ is surjective.
\end{lem}
\begin{proof}
$P/M$ is torsion-free by 3. in Proposition \ref{prop: quasi-sym}, so $\Ext^1(\overline{P/M},R)$ is trivial (see~Proposition~\ref{prop:Ext_bounds}). Therefore, the short exact sequence $0 \to M \to P \to P/M \to 0$ remains exact upon applying $*$. 
\end{proof}

\begin{prop} \label{prop:dim1_Witt_Haah}
\begin{enumerate}
    \item We have $W^s(\ZZ_{p^r}[x^{\pm 1}]) = W^s(\mathbb{F}_p[x^{\pm 1}]) =0$. That is, every symplectic module is metabolic in dimension one.
    \item Every quasi-symplectic module over $\mathbb{Z}_{p^r}[x^{\pm 1}]$ becomes metabolic upon coarse-graining. In other words, $P$ admits a Lagrangian subgroup which is invariant under a finite index subgroup of $\ZZ$.
\end{enumerate}
\end{prop}
\begin{proof}
1. $W^s(\mathbb F_p[x^{\pm 1}])=0$ is proved in \cite[Lemma IV.18]{haah2023nontrivial}. To reduce to the case of prime coefficients, as in Corollary \ref{cor:W_reduction_iso} for quasi-symplectic modules, we verify the following assertion is valid for $d=1$: if $P$ is a~symplectic module and $P'$ is constructed as in the proof of Theorem \ref{thm:reduction_to_prime}, then $P'$ is also symplectic. This follows from Lemmas \ref{lem:reduction_Witt_equivalent} and \ref{lem:surjective_res}.

2. Follows from Corollary \ref{cor:W_reduction_iso} and \cite[Corollary IV.22]{haah2023nontrivial}.
\end{proof}

    \subsection{Dimension one: associated quadratic form}

    In this subsection, we specialize to the one-dimensional case, i.e. quasi-symplectic
    modules over the Laurent polynomial ring $R = \mathbb{Z}_{n}[x^{\pm 1}]$. For
    each such module $(P,\Omega)$ we will construct a nondegenerate quadratic
    form on the group $E_P=P^* / P$, which in dimension one is always finite. This quadratic form interacts with Lagrangian submodules very well: $P$ admits a~Lagrangian submodule if and only if $E_P$ admits a~submodule which is Lagrangian for the quadratic form; see Appendix \ref{app:metric_group} for the relevant background on quadratic forms. Proposition 
    \ref{prop:dim1_Witt_Haah} shows that $P$ always admits a Lagrangian submodule after sufficient coarse-graining; therefore, $E_P$ always admits a Lagrangian subgroup, in general invariant only under a finite index subgroup of the full group of translations $\mathbb Z$. Before presenting the technical details of the construction, we~offer a brief overview.

    Starting from a quasi-symplectic module $(P, \Omega)$, we will construct
    modules $P_{\rightarrow}, P_{\leftarrow}$, and $\widehat{P}$, which
    represent nonlocal operators with support possibly extending infinitely to
    the right, to the left, or in both directions, respectively. Each contains
    $P$ as a~submodule (corresponding to operators of finite support); moreover,
    $P_{\rightarrow}, P_{\leftarrow}\subset \widehat{P}$, and their intersection
    satisfies $P_{\rightarrow}\cap P_{\leftarrow}= P$.

    As we will see, there is a perfect pairing between $P_{\rightarrow}$ and
    $P_{\leftarrow}$ obtained by extending $\Omega$. Additionally, there is a
    natural pairing between $P$ and $\widehat{P}$, which is generally degenerate:
    some elements in $\widehat{P}$ are orthogonal to all elements of $P$. These elements form a~submodule of $\widehat P$ isomorphic to $E_P$. We now outline
    the construction of this isomorphism. Every homomorphism $\alpha \in P^{*}$
    can be represented, via the duality induced by $\Omega$, by elements
    $\alpha_{\rightarrow}\in P_{\rightarrow}$ and
    $\alpha_{\leftarrow}\in P_{\leftarrow}$, see~\eqref{eq:alpha_rl} below. If~$\alpha$ represents the zero class
    in $E_P$, then $\alpha_{\rightarrow}= \alpha_{\leftarrow}\in P$. Otherwise, $\alpha_{\rightarrow}- \alpha_{\leftarrow}$ is a~nonzero
    element of $\widehat{P}$ that is orthogonal to $P$.
    
    A $\mathbb{Z}_{n}$-valued bilinear form $b$ on $E_P$ is defined by
    \begin{equation}
        b(\alpha, \beta) = \alpha_0 (\beta_{\rightarrow}- \beta_{\leftarrow}) =   \omega( \alpha_\rightarrow,\beta_\leftarrow) + \omega(\beta_\rightarrow, \alpha_\leftarrow),
        \label{eq:b_def}
    \end{equation}
    where $\alpha_0$ is the constant term of $\alpha$. Here representatives $\alpha, \beta \in P^{*}$ of their respective classes
    in $E_P$ are chosen, and $\alpha$ is identified with its natural extension
    to~$\widehat P$. The bilinear form $b$ defined in this way is symmetric and
    nondegenerate. 
    
    We also introduce the function $q : E_P \to \mathbb{Z}_{n}$,
    defined by
    \begin{equation}
        q(\alpha) = \omega(\alpha_{\rightarrow}, \alpha_{\leftarrow}).
        \label{eq:q_def}
    \end{equation}
    The function $q$ is a quadratic refinement of $b$, meaning that $q(-\alpha)=q(\alpha)$ and
    \begin{equation}
        q(\alpha + \beta) - q(\alpha) - q(\beta) = b(\alpha,\beta).
        \label{eq:qb_refinement}
    \end{equation}

    \begin{remark}
        The bilinear form $b$ can be lifted to a hermitian form $E_P \times E_P \to \widehat R$, whose image lies in the torsion submodule of $\widehat R$. This reflects the fact that $E_P$ carries a residual action of $\mathbb Z$: since $E_P$ is finite, some finite index subgroup $k \mathbb Z \subset \mathbb Z$ acts trivially on $E_P$, and $b$ is invariant under the induced action of $\mathbb Z / k \mathbb Z$. Similarly, the~quadratic form $q$ can be lifted to a map $P^* \to \widehat R$, but only the constant term of this lift is a well-defined function on $E_P$. 
        
        We prefer to work with the $\mathbb Z_n$-valued forms $b$ and $q$. Our motivation is that $E_P$, $b$ and $q$ are invariant under coarse-graining, but the residual action of translations is not: upon sufficient coarse-graining one obtains $E_P$ with a~trivial $\mathbb Z$-action. 
    \end{remark}

    The study of $E_P$ and its quadratic form $q$ culminates in Theorem \ref{thm:EP_metabolic_criterion}, which allows to decide whether $P$ is metabolic in terms of $E_P$, which is easier to analyze. We~mention two main ideas behind this. 
    \begin{itemize}
        \item If $L \subset P$ is a Lagrangian submodule, the restriction homomorphism $P^* \to L^*$ induces a homomorphism $E_P \to Q_L = \frac{L^*}{P/L}$, whose kernel is a Lagrangian submodule of $E_P$. Therefore, the existence of a Lagrangian submodule in $E_P$ is necessary for the existence of a Lagrangian submodule in $P$.
        \item Consider the problem of extending $\Omega$ to a quasi-symplectic form on some module $P^T$ that lies between $P$ and $P^*$, i.e.\ such that $P \subset P^T \subset P^*$. Such intermediate modules correspond bijectively to submodules $T$ of $E_P$, and the desired extension exists if and only if the quadratic form $q$ vanishes on $T$. If $T$ is Lagrangian in $E_P$, then $P^T$ is symplectic, rather than merely quasi-symplectic. In this case $P^T$, and hence also $P$, admits a~Lagrangian submodule by Proposition \ref{prop:dim1_Witt_Haah}. Therefore, the existence of a Lagrangian submodule in $E_P$ is sufficient for the existence of a Lagrangian submodule in $P$. 
    \end{itemize}

    In the remainder of this subsection, we spell out in detail the constructions sketched above.

    \begin{prop}
        $E_P$ has finitely many elements.
    \end{prop}
    \begin{proof}
        All finitely generated torsion modules over $\mathbb Z_n[x^{\pm 1}]$ are finite sets.
    \end{proof}

    \begin{defn}
        We introduce
        \begin{itemize}

            \item the ring $R_{\rightarrow} = \mathbb{Z}_{n}((x))$ of formal power
                series of the form $\sum_{\lambda=N}^{\infty} c_{\lambda} x^{\lambda}$,

            \item the ring $R_{\leftarrow} = \mathbb{Z}_{n}((x^{-1}))$ of formal
                power series of the form $\sum_{\lambda=-\infty}^{N} c_{\lambda}
                x^{\lambda}$,

            \item the $R$-module $\widehat R$ of bilateral formal power series $\sum
                _{\lambda=-\infty}^{\infty} c_{\lambda} x^{\lambda}$.
        \end{itemize}
The inclusions among \( R \), \( R_\rightarrow \), \( R_\leftarrow \), and \( \widehat{R} \) are represented by the following commutative diagram:
\begin{equation}
\begin{tikzcd}[scale=1.5, row sep=2em, column sep=6em]
& R_\rightarrow \arrow[dr, hookrightarrow] & \\
R \arrow[ur, hookrightarrow] \arrow[dr, hookrightarrow] & & \widehat R \\
& R_\leftarrow \arrow[ur, hookrightarrow] &
\end{tikzcd}
\end{equation}

        For every $R$-module $M$ we define
        \begin{equation}
            M_{\rightarrow} = M \otimes_{R} R_{\rightarrow}, \qquad M_{\leftarrow}
            = M \otimes_{R} R_{\leftarrow}, \qquad \widehat M = M \otimes_{R} \widehat
            R.
        \end{equation}
        Tensoring induces from $\Omega$ the pairings
        \begin{equation}
            P_{\leftarrow} \times P_{\rightarrow} \to R_{\rightarrow}, \qquad P_{\rightarrow}
            \times P_{\leftarrow} \to R_{\leftarrow}, \qquad P \times \widehat P
            \to \widehat R, \qquad \widehat P \times P \to \widehat R. \label{eq:left_right_hat}
        \end{equation}
        We denote all of them by $\Omega$, and their constant terms by
        $\omega$.

        We highlight that the involution $f \mapsto \overline f$ takes $R_\rightarrow$ to $R_\leftarrow$. Hence if $M$ is an $R_\rightarrow$-module, the involution-twisted $\overline M$ is an $R_\leftarrow$-module. 
    \end{defn}

    \begin{lem} \label{lem:1d_rings}
            Let $S$ be $R_{\rightarrow}$ or $R_{\leftarrow}$. The only
                ideals in $S$ are the principal ideals~$(m)$, where $m$ are the integers dividing $n$. $S$ is an injective module over itself and a~flat module over $R$.
    \end{lem}
    \begin{proof}
        The claimed classification of ideals in $S$ is obtained using the standard method of inverting formal series. Then self-injectivity of $S$ follows from Baer's test. The same argument applies to $K$, the total ring of fractions of $R$. Hence a~$K$-module is projective if and only if it is injective. $S$ is injective over $K$ by Baer's test; hence it it is projective, and in particular flat over $K$. $K$ is a~localization of $R$, so it is flat over $R$. Thus $S$ is flat over $R$.
    \end{proof}

    \begin{prop}
    Let $M \subset P$ be a submodule.
        \begin{enumerate}
            \item The first two pairings in \eqref{eq:left_right_hat} are perfect; that is, they induce isomorphisms
                \begin{equation}
                    P_{\rightarrow} \cong \Hom_{R_\rightarrow}(\overline{P_\leftarrow}
                    ,R_{\rightarrow}), \qquad P_{\leftarrow} \cong \Hom_{R_\leftarrow}
                    (\overline{P_\rightarrow},R_{\leftarrow}). \label{eq:perfect_pairings_1d}
                \end{equation}
                \item The canonical maps $M_\rightarrow \to P_\rightarrow$ and $M_\leftarrow \to P_\leftarrow$ are injective. Regarding them as submodule inclusions, their orthogonal complements in $P_\leftarrow$ and $P_\rightarrow$ are $(M^\perp)_\leftarrow$ and $(M^\perp)_\rightarrow$. This leads to the isomorphisms:
                \begin{subequations}
                \begin{align}
&                  (P/M^\perp)_\rightarrow \cong P_\rightarrow / (M^\perp)_\rightarrow \cong \Hom_{R_\rightarrow}(\overline{M_\leftarrow},R_\rightarrow), \\  &(P/M^\perp)_\leftarrow \cong P_\leftarrow / (M^\perp)_\leftarrow \cong \Hom_{R_\leftarrow}(\overline{M_\rightarrow},R_\leftarrow).
                \end{align}
                \end{subequations}
                \item The canonical maps in the commutative diagram
                \begin{equation}
\begin{tikzcd}[scale=1.5, row sep=2em, column sep=6em]
& M_\rightarrow \arrow[dr] & \\
M \arrow[ur] \arrow[dr] & & \widehat M \\
& M_\leftarrow \arrow[ur] &
\end{tikzcd}
\end{equation}
are injective.
    \item We have $M_\rightarrow = (M^{\perp \perp})_\rightarrow$ and $M_\leftarrow = (M^{\perp \perp})_\leftarrow$ as submodules of $P_\rightarrow$ and $P_\leftarrow$.
    \item If $M = M^{\perp \perp}$, the canonical map $\widehat M \to \widehat P$ is injective. 
        \end{enumerate}
    \end{prop}
    \begin{proof}
        1. Since $E_P$ is a torsion $R$-module and every regular element in $R$ is invertible in $R_\rightarrow$, $E_P$ is killed by $\otimes_R R_\rightarrow$: $(E_P)_\rightarrow =0$.
        In combination with the flatness of $R_{\rightarrow}$, this shows that the embedding $P
        \hookrightarrow P^{*}$ induced by $\Omega$ becomes an isomorphism after tensoring
        with $R_{\rightarrow}$:
        \begin{equation}
            P_{\rightarrow}\cong P^* \otimes_{R} R_{\rightarrow}.
        \end{equation}
        Using the flatness of $R_{\rightarrow}$ over $R$ again
        \cite[Section 2.9]{Bourbaki},
        \begin{equation}
            P^* \otimes_{R} R_{\rightarrow}\cong \Hom_{R}(\overline
            P, R_{\rightarrow}) \cong \Hom_{R_{\rightarrow}}(\overline{P_{\leftarrow}}
            ,R_{\rightarrow}).
        \end{equation}

        2. The canonical map $M_\rightarrow \to P_\rightarrow$ is injective by flatness of $R_\rightarrow$. By 2. in Proposition \ref{prop: quasi-sym}, the homology of the complex
        \begin{equation}
            0 \to M^\perp \to P \to M^* \to 0
        \end{equation}
        is torsion, and hence killed by tensoring with $R_\rightarrow$. We continue as in the proof of 1. 

        3. First consider the case $M =P$. If $p \in P$ is in the kernel of the natural map $P \to P_\rightarrow$, then it is orthogonal to $P$ and hence it equals zero. By the duality \eqref{eq:perfect_pairings_1d}, the same argument applies to $p \in P_\rightarrow$ in the kernel of the natural map $P_\rightarrow \to \widehat P$. 
        
        The case of general $M$ reduces to the special case above using 2.; for example, an~inspection of the commutative diagram
       \begin{equation}
\begin{tikzcd}[scale=1.5, row sep=4em, column sep=4em]
M_\rightarrow \arrow[r] \arrow[d, hookrightarrow] & \widehat M \arrow[d] \\
P_\rightarrow \arrow[r, hookrightarrow] & \widehat P
\end{tikzcd}
\end{equation}
immediately shows that the horizontal arrow in the top row is injective.

     4. $M^{\perp \perp}/ M$ is a torsion module by 3. in Proposition \ref{prop: quasi-sym}, so $\otimes_R R_\rightarrow$ kills it. We~invoke, once again, the flatness of $R_\rightarrow$.

        The same arguments in 1.-4. can be repeated with the roles of $\rightarrow$
        and $\leftarrow$ switched throughout .

        5. If $M = M^{\perp \perp}$, then $P/M$ is torsion-free by 3. in Proposition \ref{prop: quasi-sym}. Hence $\Tor_1^R(P/M,\widehat R)=0$. 
    \end{proof}

    Let \( \alpha \in P^* \). By tensoring with \( \mathrm{id}_{\widehat{R}} \), we obtain an $R$-antilinear extension \( \widehat{P} \to \widehat{R} \), which, by a slight abuse of notation, we continue to denote by~\( \alpha \). Functoriality of the tensor product implies that $\alpha$ takes $P_\leftarrow$ to $R_\rightarrow$ and $P_\rightarrow$ to $R_\leftarrow$, so by the duality \eqref{eq:perfect_pairings_1d} there exist unique $\alpha_\rightarrow \in P_\rightarrow$ and $\alpha_\leftarrow \in P_\leftarrow$ such that 
    \begin{subequations}
    \label{eq:alpha_rl}
    \begin{align}
        \alpha(p) &= \Omega(p, \alpha_\rightarrow) \qquad \text{for } p\in P_\leftarrow, \\
        \alpha(p) &= \Omega(p, \alpha_\leftarrow) \qquad \text{for } p\in P_\rightarrow.
    \end{align}
    \end{subequations}
    It follows immediately that
    \begin{equation}
     \Omega(p, \alpha_\rightarrow - \alpha_\leftarrow) = 0 \qquad \text{for } p \in P.   
    \end{equation}

    \begin{prop} \label{prop:finite_line_criterion} 
    \begin{enumerate}
        \item Let $p \in P_{\rightarrow}$ and let $M \subset P$ be a submodule. Then:
        \begin{enumerate}
            \item $p$ can be represented as $\alpha_{\rightarrow}$ for some $\alpha \in P^*$ if and only
                if $\Omega(q,p) \in R$ for all $q \in P$.
            \item $p \in P+ M_\rightarrow$ if and only if $\Omega
                (\alpha_{\leftarrow},p ) \in R$ for all $\alpha \in P^*$ vanishing on $M$.
        \end{enumerate}
        The same equivalences holds with $\rightarrow$ and
        $\leftarrow$ interchanged.
        \item The intersection $P_\rightarrow \cap P_\leftarrow $ in $\widehat P$ equals $P$.
    \end{enumerate}
    \end{prop}
    \begin{proof}
        1. (a) If $p = \alpha_{\rightarrow}$, then $\Omega(q,p) = \alpha(q) \in R$ holds for $q \in P$. Conversely, define $\alpha$ by $\alpha(q) = \Omega(q,p)$.
        By the uniqueness of $\alpha_{\rightarrow}$, $p = \alpha_{\rightarrow}$.

        1. (b) 
        $\implies$ is immediate, and we will prove $\impliedby$. The module of $\alpha \in P^*$ vanishing on $M$ can be identified with $(P/M)^*$. Define $\phi \in (P/M)^{**}$ by
        \begin{equation}
            \phi(\alpha)= \Omega(\alpha_{\leftarrow},p).
            \label{eq:phi_double_dual_argument}
        \end{equation}
        Since $\mathrm{Ext}^{2}_{R}(\cdot,R)=0$, the canonical map 
        $P/M \to (P/M)^{**}$ is surjective
        \cite[Proposition 5]{Masiek}. That is, there exists $q \in P$ such that
        \begin{equation}
            \phi(\alpha)= \overline{\alpha(q)}= \overline{\Omega(q,\alpha_\leftarrow)}
            = - \Omega(\alpha_{\leftarrow},q).
        \end{equation}
        Comparing this with \eqref{eq:phi_double_dual_argument} we see that $p+q$ is orthogonal to $\alpha_\leftarrow$ for all $\alpha \in (P/M)^*$. In particular, $p+q$ is orthogonal to $M^\perp$, so it belongs to $(M^{\perp \perp})_\rightarrow = M_\rightarrow$.
    
2. \( P \subset P_\rightarrow \cap P_\leftarrow \) is evident. We prove $\supset$. Let $q \in P_\rightarrow \cap P_\leftarrow$. We~have $q \in P_\rightarrow $, so~we can apply the criterion 1. (b) with $M=0$ to $q$. If~$\alpha \in P^*$, then
\begin{equation}
   R_\rightarrow \ni  \Omega(\alpha_\leftarrow, q ) = - \overline{\alpha (q) } = \Omega(\alpha_{\rightarrow},q) \in R_{\leftarrow},
\end{equation}
where in the second equality we used $q \in P_{\leftarrow}$. Since $R_\rightarrow \cap R_\leftarrow = R$, we have $\Omega(\alpha_\leftarrow, q ) \in R$ and hence $q \in P$.
\end{proof}

\begin{thm}
\begin{enumerate}
    \item There is an isomorphism 
    \begin{align}
    E_P & \xrightarrow[]{\sim} \{ m\in \widehat P \, | \, \Omega(m, p)=0 \text{ for all } p \in P \}, \nonumber \\
  \text{given by:} \qquad  \alpha &\mapsto \alpha_\rightarrow - \alpha_\leftarrow.
    \label{eq:functional_to_line}
\end{align}
\item If $M \subset P$ is a submodule and $\alpha \in P^*$, then $\alpha_\rightarrow - \alpha_\leftarrow \in \widehat{M^{\perp}}$ if and only if the class of $\alpha \in E_P$ admits a representative vanishing on $M$. 
\end{enumerate}
\end{thm}
\begin{proof}
1. If $\alpha$ corresponds to an element of $P$, i.e.\ $\alpha(\cdot) = \Omega(\cdot,p)$ for some $p \in P$, then $\alpha_\rightarrow = \alpha_\leftarrow = p$. Conversely, if $\alpha_\rightarrow = \alpha_\leftarrow$, then $\alpha_{\rightarrow} \in P_\rightarrow \cap P_\leftarrow = P$ (by 2. in Proposition \ref{prop:finite_line_criterion}). Hence $\alpha(\cdot) = \Omega(\cdot,p)$ with $p = \alpha_\rightarrow$. 

This equivalence proven above shows that the map \eqref{eq:functional_to_line} naturally factors through an injective map $E_P \to \widehat P$. Its image is contained in the orthogonal complement of~$P$, as seen from \eqref{eq:alpha_rl}. We have to prove that all elements in this orthogonal complement lie in the image of \eqref{eq:functional_to_line}. 

Let $q \in \widehat P$ be such that $\Omega(\cdot, q ) =0$. Since $\widehat P = P_\rightarrow + P_\leftarrow$, there exist $q_\rightarrow \in P_\rightarrow$ and $q_\leftarrow \in P_\leftarrow$ such that $q = q_\rightarrow-q_\leftarrow$. The assumption on $q$ implies that for every $p \in P$ we have
\begin{equation}
    R_\leftarrow \ni \Omega(p,q_\leftarrow) = \Omega(p,q_\rightarrow) \in R_\rightarrow.
\end{equation}
Therefore $\alpha (\cdot) = \Omega(\cdot, q_\rightarrow) \in P^*$. Clearly $\alpha_\rightarrow = q_\rightarrow$ and $\alpha_\leftarrow = q_\leftarrow$, so we have $q = \alpha_\rightarrow - \alpha_\leftarrow$ as desired. 

2. First suppose that $\alpha$ vanishes on $M$. Then $\alpha_\rightarrow \in (M^{\perp})_\rightarrow$ and $\alpha_\leftarrow \in (M^{\perp})_\leftarrow$, so $\alpha_\rightarrow - \alpha_\leftarrow \in \widehat{M^{\perp}}$. 

For the converse, assume that $\alpha_\rightarrow - \alpha_\leftarrow \in \widehat{M^{\perp}} =  (M^{\perp})_\rightarrow + (M^\perp)_\leftarrow$. There exist $p_\rightarrow \in (M^{\perp})_\rightarrow$ and $p_\leftarrow \in (M^\perp)_\leftarrow$ such that 
\begin{equation}
    \alpha_\rightarrow - p_\rightarrow = \alpha_\leftarrow - p_\leftarrow.
\end{equation}
Since $P_\rightarrow \cap P_\leftarrow = P$, this implies that the above equality takes place in $P$. Thus we can replace $\alpha \in P^*$ with another one, representing the same class in $E_P$, so that
\begin{equation}
    \alpha_\rightarrow = p_\rightarrow, \qquad \alpha_\leftarrow = p_\leftarrow.
\end{equation}
Then $\alpha$ vanishes on $M$.
\end{proof}

\begin{thm}
    The bilinear form $b$ in \eqref{eq:b_def} is well-defined on $E_P$, symmetric, and~nondegenerate in the sense that
    \begin{equation}
        E_P \ni \alpha \mapsto b(\alpha, \cdot) \in \mathrm{Hom}_{\mathbb Z_n}(E_P, \mathbb Z_n)
        \label{eq:b_perfect_duality}
    \end{equation}
    is a bijection. The function $q$ in \eqref{eq:q_def} is well-defined on $E_P$ and is a quadratic refinement of $b$; that is, \eqref{eq:qb_refinement} holds.
\end{thm}
\begin{proof}
Consider $b$ first as a bilinear form on $P^*$. If $\beta$ represents the trivial class in~$E_P$, then $\beta_\rightarrow = \beta_\leftarrow$, and the vanishing of $b(\alpha,\beta)$ is seen from the first expression in \eqref{eq:b_def}. The second formula for $b(\alpha,\beta)$ in \eqref{eq:b_def} shows that $b$ is symmetric, so it vanishes also if $\alpha$ represents the trivial class. This ensures that $b$ descends to a~bilinear form on~$E_P$.  

Now suppose that $b(\alpha,\cdot)=0$. Then for every $\beta \in P^*$,
\begin{equation}
R_\leftarrow \ni    \Omega(\beta_\rightarrow,\alpha_\leftarrow) = \Omega(\beta_\leftarrow, \alpha_\rightarrow) \in R_\rightarrow.
\label{eq:b_symmetric_formula}
\end{equation}
Hence $\Omega(\beta_\leftarrow, \alpha_\rightarrow) \in R$. Applying Proposition \ref{prop:finite_line_criterion}, 1. (b) to $p = \alpha_\rightarrow$, we conclude that $\alpha_\rightarrow \in P$. Therefore, $\alpha$ represents the trivial class in $E_P$.

We showed that the map in \eqref{eq:b_perfect_duality} is injective. Since $E_P$ is a finite $\mathbb Z_n$-module, it~has as many elements as $\Hom_{\mathbb Z_n}(E_P,\mathbb Z_n)$. Therefore, \eqref{eq:b_perfect_duality} is a bijection. 

To verify that $q$ is well defined on $E_P$, suppose that $\alpha, \alpha' \in P^*$ represent the same class. There exists $p \in P$ such that $\alpha'_\rightarrow = \alpha_\rightarrow + p$ and $\alpha'_\leftarrow = \alpha_\leftarrow +p$. Hence
\begin{equation}
q(\alpha') = \omega(\alpha_\rightarrow + p, \alpha_\leftarrow +p) = q(\alpha) + \alpha_0(p) - \alpha_0(p) + \omega(p,p) = q(\alpha).
\end{equation}
The quadratic refinement property \eqref{eq:qb_refinement} is immediate from \eqref{eq:b_symmetric_formula}.
\end{proof}

\begin{prop} \label{prop:Lag_induces_Lag}
    Let $M \subset P$ be a submodule, and let $Q_M = \frac{M^*}{P/M^{\perp}}$ be the cokernel of the map in \eqref{eq:quasi_symp_cokernel}. Let 
    \begin{equation}
\rho_M : E_P \to Q_M
    \end{equation}
    be the homomorphism induced by the restriction map $P^* \to M^*$.
    \begin{enumerate}
        \item If $M=M^{\perp \perp}$, then $\rho_M$ is surjective.
        \item $\ker(\rho_M)$ is the submodule of $E_P$ consisting of these classes which admit a~representative vanishing on $M$. 
        \item If $M^\perp$ is isotropic, then $\ker(\rho_M)$ is isotropic in $E_P$: $q$ vanishes on $\ker(\rho_M)$, and in particular $\ker(\rho_M) \subset \ker(\rho_M)^\perp$.
        \item We have $\ker(\rho_M)^\perp = \ker(\rho_{M^{\perp}})$.
        \item If $M$ is Lagrangian, then $\ker(\rho_M)$ is Lagrangian in $E_P$: $q$ vanishes on $\ker(\rho_M)$ and $\ker(\rho_M) = \ker(\rho_M)^\perp$. 
        \item If $M=M^{\perp \perp}$ and $M^\perp$ is isotropic, the quasi-symplectic module $M/M^{\perp \perp}$ satisfies $E_{M/M^{\perp}} \cong \frac{\ker(\rho_M)^\perp}{\ker(\rho_M)}$.
    \end{enumerate}
\end{prop}
\begin{proof}
    1. See Lemma \ref{lem:surjective_res}.

    2. We have the commutative diagram with exact rows
    \begin{equation}
\begin{tikzcd}[row sep=3em, column sep=3.5em]
0 \arrow[r] & M^{\perp} \arrow[r] \arrow[d, hookrightarrow] & P \arrow[r] \arrow[d, hookrightarrow] & P/M^{\perp} \arrow[r] \arrow[d, hookrightarrow] & 0 \\
0 \arrow[r] & (P/M)^* \arrow[r] & P^* \arrow[r] & M^* &
\end{tikzcd}
\end{equation}
in which the bottom row is obtained by applying the functor $*$ to the short exact sequence $0 \to M \to P \to P/M \to 0$, $(P/M)^*$ is identified with the elements of $P^*$ vanishing on $M$ and the vertical arrows are given by $\Omega$. The snake lemma provides the exact sequence
\begin{equation}
    0 \to \frac{(P/M)^*}{M^\perp} \to E_P \xrightarrow{\rho_M} Q_M.
\end{equation}

3. If $\alpha \in (P/M)^*$, then:
\begin{equation}
 \alpha_\rightarrow \in (M^{\perp})_\rightarrow, \qquad \alpha_\leftarrow \in (M^\perp)_\leftarrow   .
\end{equation}
If $M^{\perp}$ is isotropic, it follows immediately that $q(\alpha)=0$. 

4. If $\alpha \in (P/M)^*$ and $\beta \in (P/M^{\perp})^*$, then the argument from the proof of 3. shows that $b(\alpha, \beta)=0$. The description of $\ker(\rho_M)$ in 2. shows that 
\begin{equation}
    \ker(\rho_{M^\perp}) \subset \ker(\rho_M)^{\perp}.
    \label{eq:kerrho_orthogonality_inclusion}
\end{equation}
For the opposite inclusion, suppose that $\alpha \in \ker(\rho_M)^\perp$ and $\beta \in (P/M)^*$. Thus $\beta$ represents a class in $\ker(\rho_M)$. The equality $b(\alpha,\beta)=0$ holds also if $\beta$ is replaced by $x^n \beta$ for any $n \in \mathbb Z$, so we have
\begin{equation}
    R_\leftarrow \ni \Omega(\alpha_\rightarrow, \beta_\leftarrow) = \Omega(\alpha_\leftarrow, \beta_\rightarrow) \in R_\rightarrow.
\end{equation}
It follows that $\Omega(\alpha_\rightarrow,\beta_\leftarrow) \in R_\rightarrow \cap R_\leftarrow = R$. By Proposition \ref{prop:finite_line_criterion}, 2. (b), we have $\alpha_\rightarrow \in P + M_\rightarrow$. We can replace $\alpha$ by another one in the same class to ensure that $\alpha_\rightarrow \in M_\rightarrow$. Then $\alpha$ vanishes on $M^\perp$, so the class of $\alpha$ is in $\ker(\rho_{M^\perp})$. 

5. is immediate from 3. and 4.

6. Surjective restriction map $P^* \to M^*$ induces a surjection $(P/M^{\perp})^* \to (M/M^{\perp})^*$. By the previous points, we have the following commutative diagram with exact rows:
\begin{equation}
\begin{tikzcd}
0 \arrow[r] & M \arrow[r] \arrow[d, two heads] & (P/M^{\perp})^* \arrow[r] \arrow[d, two heads] & \ker(\rho_M)^\perp \arrow[r] \arrow[d, "f"] & 0 \\
0 \arrow[r] & M/M^{\perp} \arrow[r] & (M/M^{\perp})^* \arrow[r] & E_{M/M^{\perp}} \arrow[r] & 0
\end{tikzcd}
\end{equation}
An inspection of the diagram shows that the third arrow, denoted by $f$, is also surjective. We claim that $\ker(f)=\ker(\rho_M)$. In one direction, a class which admits a representative vanishing on $M$ clearly belongs to the kernel of $f$. Conversely, suppose that $\alpha \in P^*$ vanishes on $M^\perp$ and represents a class in $\ker(f)$. Chasing the diagram we see that there exists $m \in M$ such that $\left. \alpha(\cdot) \right|_M = \left. \Omega(\cdot, m ) \right|_M$. Hence $\alpha(\cdot) - \Omega(\cdot, m)$ represents the same class and vanishes on $M$; that is, the class of $\alpha$ is in $\ker(\rho_M)$.
\end{proof}

\begin{prop} \label{prop:OmegaT}
Let $T \subset E_P$ be an isotropic submodule, i.e., $\left. q \right|_T =0$. Let
\begin{equation}
    P^T = \{ \alpha \in P^* \, | \alpha \text{ represents a class in } T \},
\end{equation}
and consider $P$ as a submodule of $P^T$ via the embedding $p \mapsto \Omega(\cdot,p)$. Then $\Omega$ extends uniquely to a sesquilinear form $P^T \times P^T \to R$. This extension is given explicitly by
\begin{equation}
    \Omega^T (\alpha, \beta) = \Omega(\alpha_\leftarrow,\beta_\rightarrow), 
    \label{eq:Omega_extension}
\end{equation}
and makes $P^T$ a quasi-symplectic module. We have an isomorphism
\begin{equation}
    E_{P^T} \cong \frac{T^\perp}{T}. 
    \label{eq:EPT}
\end{equation}
In particular, if $T$ is Lagrangian in $E_P$, then $P^T$ is symplectic. 
\end{prop}
\begin{proof}
Since $P^T/P$ is torsion, there exists at most one $R$-valued sesquilinear extension of $\Omega$ to $P^T$. The formula \eqref{eq:Omega_extension} defines a sesquilinear form, a~priori valued in $R_\rightarrow$. We will show that it turns $P^T$ into a~quasi-symplectic module. 

Let $\alpha,\beta \in P^*$ represent classes in $T$. Equality $q(\alpha)=0$ implies that the constant term of $\Omega^T(\alpha,\alpha)$ equals zero. All translates of $\alpha, \beta$ are $b$-orthogonal because $T$ is a~submodule, so
\begin{equation}
R_\rightarrow \ni \Omega(\alpha_\leftarrow,\beta_\rightarrow) = \Omega(\alpha_\rightarrow,\beta_\leftarrow) \in R_\leftarrow.  
\label{eq:OmegaT_rightleft}
\end{equation}
It follows that $\Omega^T(\alpha,\beta) \in R$. Anti-hermicity of $\Omega^T$ is also readily verified with \eqref{eq:OmegaT_rightleft}, though this is actually redundant. 

To prove \eqref{eq:EPT}, first note that the restriction map $(P^T)^* \to P^*$ is injective. We~claim that its image is $P^{S}$, where $S=T^{\perp}$. Indeed, if $\beta \in P^*$, then $\alpha \mapsto \Omega(\alpha_\leftarrow,\beta_\rightarrow)$ is the unique $R_\rightarrow$-valued extension of $\beta$ to $P^T$. This extension is $R$-valued if and only if $\beta \in P^S$, by an argument analogous to that in \eqref{eq:OmegaT_rightleft}. Hence,
\begin{equation}
    (P^T)^*/P^T \cong P^S / P^T \cong S/T. 
\end{equation}
\end{proof}

\begin{thm} \label{thm:EP_metabolic_criterion}
The following conditions are equivalent:
\begin{enumerate}
    \item $E_P$ admits a Lagrangian submodule,
    \item $P$ can be embedded in a symplectic module $P'$, with $P'/P$ a torsion module,
    \item $P$ admits a Lagrangian submodule.
\end{enumerate}
Moreover, these conditions are always satisfied upon sufficient coarse-graining. In~particular, $E_P$ admits a Lagrangian subgroup.  
\end{thm}
\begin{proof}
1. $\implies$ 2. by Proposition \ref{prop:OmegaT}, and 3. $\implies$ 1. by Proposition \ref{prop:Lag_induces_Lag}. We prove that 2. $\implies$ 3. $P'$ admits a Lagrangian submodule $L'$ by Proposition \ref{prop:dim1_Witt_Haah}. Then $L = L' \cap P$ is a Lagrangian submodule of $P$.

The final claim is included in Proposition \ref{prop:dim1_Witt_Haah}.
\end{proof}

    \section{Stabilizer codes in half-spaces}
    
    \subsection{Modules of boundary operators}

    In this section we consider stabilizer codes in the geometry of a half-space,
    \begin{equation}
    \Lambda_{\geq 0} := \{a \in \mathbb{Z}^d \mid v \cdot a \geq 0\} \;\cong\; \mathbb{Z}^{d-1} \times \mathbb{N},
    \label{eq:half_space}
    \end{equation}
    where $v \in \mathbb{Z}^d$ is a fixed nonzero vector. This geometry does not possess the full $\mathbb Z^d$ translation symmetry, but one may still ask for stabilizer codes invariant under translations parallel to the boundary of $\Lambda_{\geq 0}$, i.e.\ perpendicular to $v$. Such translations form a group isomorphic to~$\mathbb Z^{d-1}$, and the corresponding group ring $R_\partial$ is a Laurent polynomial ring in $d-1$ variables.

    Let $R = \mathbb Z_n[x_1^{\pm 1},\dots,x_d^{\pm 1}]$ and let $P \cong R^{2q} \cong \bigoplus_{\lambda \in \mathbb Z^d} \mathbb Z_n^{2q}$ be a standard symplectic module describing Pauli operators in $\mathbb Z^d$. Then
    \begin{equation}
        P_{\geq 0} = \bigoplus_{\lambda \in \Lambda_{\geq 0}} \mathbb Z_n^{2q} \subset P
        \label{eq:P_halfspace}
    \end{equation}  
    describes Pauli operators in $\Lambda_{\geq 0}$. Note that $P_{\geq 0}$ is an $R_\partial$-submodule of $P$, but not an $R$-submodule. It is not finitely generated over $R_\partial$.

    Let $L \subset P$ be a Lagrangian $R$-submodule, corresponding to a stabilizer code in the full space $\mathbb Z^d$. We are interested in constructing a corresponding stabilizer code in the half-space geometry; that is, a Lagrangian subgroup of $P_{\geq 0}$ with respect to the restricted form $\left. \omega \right|_{P_{\geq 0}}$. We would like to preserve translation symmetry parallel to the boundary, so we look for Lagrangian $R_{\partial}$-submodules.

    It is a natural requirement that the half-space Lagrangian contains
    \begin{equation}
        L_{\geq 0} := L \cap P_{\geq 0},
    \end{equation}
    sometimes called the group of bulk stabilizers. It is apparent that $L_{\geq 0}$ is isotropic, but typically not Lagrangian in $P_{\geq 0}$. Completing it to a Lagrangian submodule requires adjoining additional stabilizers acting near the boundary of $\Lambda_{\geq 0}$; these can be interpreted as specifying boundary conditions. In general, there are many possible choices of boundary conditions, but it is not always possible to choose translation-invariant boundary conditions. 

    Every Lagrangian subgroup containing $L_{\geq 0}$ is contained in $L_{\geq 0}^\perp \cap P_{\geq 0}$, the orthogonal complement of $L_{\geq 0}$ within $P_{\geq 0}$. Since groups $M$ satisfying $L_{\geq 0} \subset M \subset L_{\geq 0}^\perp$ are canonically in a one-to-one correspondence with subgroups of $L_{\geq 0}^\perp / L_{\geq 0}$, we are naturally led to the following definition.
    
    \begin{defn}
The \textbf{boundary operator module} is defined as
\begin{equation}
    P_{\partial} := \left( L_{\geq 0}^{\perp} \cap P_{\geq 0} \right) / L_{\geq 0}.
\end{equation}

 The form $\omega: P\times P\rightarrow \mathbb{Z}_{n}$ descends to an alternating
    form $\omega_{\partial}: P_{\partial}\times P_{\partial} \rightarrow \mathbb{Z}
    _{n}$: given two classes
    $p_{1}+L_{\geq 0}, p_{2}+L_{\geq 0}\in P_{\partial},$
    \begin{equation}
        \omega_{\partial}(p_{1}+L_{\geq 0}, p_{2}+L_{\geq 0}):=\omega(p_{1}, p_{2}
        ).
    \end{equation}
    
        This form $\omega_{\partial}$ encodes the commutation relations of Pauli operators supported in $\Lambda_{\geq 0}$ and commuting with bulk stabilizers. It can be extended uniquely to an $R_\partial$-sesquilinear form $\Omega_\partial$ whose constant term is $\omega_\partial$.
\end{defn}

 It will be shown below that $P_{\partial}$ is a quasi-symplectic $R_\partial$-module. It describes degrees of freedom localized near the boundary of $\Lambda_{\geq 0}$. By construction, Lagrangian subgroups of $P_{\geq 0}$ containing $L_{\geq 0}$ correspond to Lagrangian subgroups of~$P_{\partial}$. This raises the question whether $P_\partial$ admits a Lagrangian submodule. A~negative answer is already possible in dimension $d=2$, as illustrated by Wen's plaquette model \cite{WenPlaquette}. However, the theory of Witt groups developed in Section~\ref{sec:weak_witt} shows that for $d=2$, corresponding to one-dimensional boundary, a~Lagrangian submodule of $P_\partial$ always exists after sufficient coarse-graining. In~other words, there exist boundary conditions invariant under a finite index subgroup of the group of translations parallel to the boundary.

 We will now provide details of the constructions outlined above. First, let us review some facts about projections in
    Pauli modules. 
    Given $S \subset \mathbb{Z}^{d}$,
    $\pi_{S}: P \rightarrow P$ denotes the projection onto $P_{S} \cong \bigoplus_{\lambda \in S} \mathbb Z_n^{2q}$, the subgroup
    of Pauli operators (modulo phase) supported on $S$. This projection acts by
    setting to zero all components outside $S$. 
    \begin{lem} \label{lem: orthogonal}
        For any $S\subset \ZZ^{d}$, $\pi=\pi_{S}: P\rightarrow P$ is a group
        homomorphism satisfying:
        \begin{enumerate}
            \item (Projection) $\pi\circ \pi=\pi$, and,

            \item (Orthogonality) $\omega(\ker \pi, \pi P)=0$.
        \end{enumerate}
    \end{lem}
 
         \begin{lem} \label{lem: double_perp}
                 Given a Lagrangian subgroup $L\subset P$ and a map $\pi : P \to P$ with the properties 1. and 2. stated in Lemma~\ref{lem: orthogonal}, let $L_{\pi}:= L\cap \pi P$. Then:
        \begin{enumerate}
            \item $L_{\pi}=(\pi L)^{\perp}\cap \pi P ,$
            \item $(\pi L)^{\perp\perp}= L_{\pi}^{\perp}\cap \pi P,$
        \item $L_{\pi}^{\perp\perp}
        =L_{\pi}.$
        \end{enumerate}
        
         \end{lem}

    \begin{proof}
    For a subgroup $M \subset \pi P$, $M^\perp \cap \pi P$ is its orthogonal complement within $\pi P$. We do not write simply $M^\perp$ to avoid confusion with the orthogonal complement in~$P$. We note two standard properties. Firstly, the three-fold orthogonal complement in $\pi P$ coincides with the orthogonal complement. Secondly, double complement of $M \subset \pi P$ within $\pi P$ equals the double complement in $P$.  
        \begin{enumerate}
            \item Given $x\in L_{\pi}$ and $y\in \pi L$, choose $l\in L$ with
        $\pi l= y$. Then 
        \begin{equation}
         \omega(x,y)=\omega(x,l)+\omega(x,y-l)=0,   
        \end{equation}
        where the first
        term vanishes because $x, l \in L$ and the second term vanishes
        by orthogonality of $\pi$. This proves $\subset$. Conversely, if $x \in (\pi L)^\perp \cap \pi P$, then for every $l \in L$
        \begin{equation}
            \omega(x,l) = \omega(\pi x,l) = \omega(x, \pi l ) = 0.
        \end{equation}
        Hence $x \in L^\perp = L$. 
        \item From 1., $ L_\pi^\perp \cap \pi P= \left((\pi L)^{\perp}\cap \pi P\right)^\perp \cap \pi P=(\pi L)^{\perp\perp}.$
        \item From 1., $L_\pi^{\perp\perp}=\left((\pi L)^{\perp}\cap \pi P\right)^{\perp\perp}=(\pi L)^{\perp \perp \perp}\cap \pi P=(\pi L)^\perp \cap \pi P=L_\pi$.
        \end{enumerate}
    \end{proof}

The group $P_{\geq 0}$ defined in \eqref{eq:P_halfspace} is the image of the orthogonal projection $\pi_{\geq 0} $ associated with the half-space $\Lambda_{\geq 0}$. Similarly, we let $P_{<0}$ to be the image of the projection $\pi_{<0} =1 - \pi_{\geq 0}$, corresponding to
\begin{equation}
\Lambda_{<0} := \{a \in \mathbb{Z}^d \mid v \cdot a < 0\}.
\end{equation}
Both $P_{\geq 0}$ and $P_{<0}$ are (infinitely generated) modules over the subring $R_{\partial} \subset R$, which is generated by monomials $x^\lambda$ with exponents $\lambda$ perpendicular to~$v$. 

The remainder of this section assumes the following convention:
\begin{equation}
    \Lambda_{\geq 0}:=\{(a_{1}, \dots, a_{d})\in \ZZ^d:  a_d\geq 0\},
    \label{eq:standard_halfspace}
\end{equation} and 
\begin{equation}
    R_{\partial}:=\mathbb{Z}_{n}[x_{1}
    ^{\pm}, \dots, x_{d-1}^{\pm}], \qquad R= R_\partial [x_d^{\pm}].
\end{equation}
This amounts to choosing a basis of $\mathbb Z^d$ in which the first $d-1$ vectors are orthogonal to $v$. Such choice is always possible. We hope that the use of generators which depend on the half-space will not lead to confusion in discussion where different half-spaces are compared.



In view of the discussion at the end of Section \ref{sec:quasi-symp}, Lemma \ref{lem: double_perp} exhibits $L_\pi^\perp \cap \pi P$ as the topological closure of $\pi L$. In general the closure is really needed, i.e.\ $L_\pi^\perp \cap \pi P$ may strictly contain $\pi L$. By contrast, $L_\pi$ is always closed. Indeed, it is defined as the intersection of two closed subgroups, $L$ and $\pi P$. This discussion motivates the following definition, which appears also in \cite{liang2024operator}.

\begin{defn}
The \textbf{primary boundary operator module} is the submodule
\begin{equation}
    P^\mathrm{pr}_{\partial} := \pi_{\geq 0}L / L_{\geq 0} \subseteq P_{\partial},
\end{equation}
with an analogous definition for the opposite half-space. Elements of $P_{\partial} \setminus P^\mathrm{pr}_{\partial}$ are called \textbf{secondary boundary operators}.
\end{defn}

\begin{remark}
 Let us record some remarks about secondary operators.
\begin{itemize}
    \item The existence of secondary boundary operators depends on the choice of the half-space $\Lambda_{\geq 0}$, and not solely on the stabilizer code.
    
    \item The saturation of $P_{\partial}^{\mathrm{pr}}$ in $P_{\partial}$ is all of $P_\partial$: for every $p \in P_\partial$, there exists a~regular element $r \in R_\partial$ and an element $l \in L$ such that $\pi_{\geq 0}(l)$ represents $rp$; see~Proposition~\ref{prop: quasi-sym}.
    
    \item Every $\alpha \in P_{\partial}^* = \mathrm{Hom}_{R_{\partial}}(\overline{P_\partial}, R_{\partial})$ is uniquely determined by its restriction to $P_{\partial}^{\mathrm{pr}}$. However, not every $\alpha \in (P_{\partial}^{\mathrm{pr}})^*$ extends to  $P_\partial^*$. In other words, we have an injection $P_\partial^* \to (P_\partial^{\mathrm{pr}})^*$, in general not surjective.
    
    \item The modules $P_\partial$ and $P_\partial^{\mathrm{pr}}$ are Witt equivalent by Lemma~\ref{lem:reduction_Witt_equivalent}.
    
    \item The extended modules $E_{P_\partial}$ and $E_{P_{\partial}^{\mathrm{pr}}}$ need not be equal. In the example below, $P_\partial$ is symplectic while $E_{P_{\partial}^{\mathrm{pr}}} \neq 0$.
\end{itemize}
   
\end{remark}

\begin{exmp}
Let $P = \mathbb Z_2[x_1^{\pm 1}, x_2^{\pm 1}]$ and let $L$ be the image of
\begin{equation}
    \sigma = \begin{pmatrix}
        (1 + \overline x_1) \overline x_2 + (1+x_1) x_2 \\ 1
    \end{pmatrix}.
\end{equation}
We consider the standard half-space geometry \eqref{eq:standard_halfspace}. Then $\pi_{\geq 0} L$ is generated over $R_{\partial} = \mathbb Z_2[x_1^{\pm 1}]$ by $L_{\geq 0}$ and the elements
\begin{equation}
    e_1 = \begin{pmatrix}
        1+x_1 \\ 0 
    \end{pmatrix}, \qquad e_2 = \begin{pmatrix}
        (1+x_1) x_2 \\ 1
    \end{pmatrix}.
\end{equation}
The module $P_\partial^{\mathrm{pr}}$ is freely generated by $e_1,e_2$, which satisfy
\begin{equation}
    \Omega_\partial(e_1,e_2) = 1 + \overline x_1.
\end{equation}
From this one finds that $E_{P_{\partial}^{\mathrm{pr}}} \cong (R_{\partial} / (x+1))^{\oplus 2} \neq 0$. 

Now let us note that
\begin{equation}
    \sum_{j=0}^{N-1} x_1^j e_1 = \begin{pmatrix}
        1 + x_1^N \\ 0 
    \end{pmatrix} \xrightarrow{N \to \infty} \begin{pmatrix}
        1 \\ 0
    \end{pmatrix} =:  e_1'.
\end{equation}
The limit vector $ e_1'$ is a secondary satisfying $(1+x_1)e_1'=e_1$, and together with $e_2$ it generates $P_\partial$. We~have
\begin{equation}
    \Omega_{\partial}( e_1',e_2) = 1,
\end{equation}
so $P_\partial$ is a standard symplectic module.
\end{exmp}
See \cite{liang2024operator, schuster2023holographic} for more examples of boundary calculations.

        We now state a fundamental result about the structure of $(P_{\partial}, \Omega_\partial)$.

    \begin{thm} \label{thm: fundamental}
        The $R_{\partial}$-module
        $P_{\partial}$ equipped with $\Omega_{\partial}$ is quasi-symplectic.
    \end{thm}

    \begin{proof}
        The only property of $(P_\partial,\Omega_\partial)$ which requires an argument is that $P_{\partial}$ is finitely generated over $R_\partial$.
        Choose a finite set $\{l_{1}, l_{2}, \dots, l_{k}\}$ of $R$-generators for
        $L$. Notice that $\pi_{\geq 0}L$ is generated over $R_{\partial}$ by the infinite
        set
        $\{\pi_{\geq 0}x_{d}^{i_1}l_{1}, \dots, \pi_{\geq 0}x_{d}^{i_k}l_{k}: i_{1}
        , \dots, i_{k} \in \ZZ\}$. Since each generator $l_i$ is supported in a~bounded region, there exist\footnote{The quantity $M_i - m_i$ is the geometric width of the support of $l_i$.}
        $(m_{1}, M_{1}), (m_{2}, M_{2}), \dots, (m_{k}, M_{k})\in \ZZ^{2}$ such
        that for any $1 \leq j \leq k$,
        \begin{align}
            \pi_{\geq 0}x_{d}^{i_j}l_{j} = \begin{cases}
                0, & i_j< m_j,\\ 
                \pi_{\geq 0} x_d^{i_j}l_j, & m_j \leq i_j\leq M_j, \\
                x_d^{i_j}l_j, & i_j> M_j.
            \end{cases}
        \end{align}
                All but finitely of them (the second case) are contained in $L_{\geq 0}$.                 In~particular, there exists $M \geq 0$ such that every $l \in \pi_{\geq 0}L$ can be decomposed as 
                \begin{equation}
                 l = \tilde{l} + l', 
                \end{equation}
                with $l' \in L_{\geq 0}$ and $\tilde{l} \in P_{\Lambda_M}$, where
                \begin{equation}
                    \Lambda_{M} := \{ (\mathbf{m}, h) : \mathbf{m} \in \mathbb{Z}^{d-1},\ h \in \{0, 1, \dots, M\} \}.
                \end{equation}
        Equivalently, $\pi_{\geq 0}L$ is contained in $P_{\Lambda_M} + L_{\geq 0}$. For any $x \in L_{\geq 0}^{\perp} \cap P_{\geq 0}$ and $l \in \pi_{\geq 0}L$, we have
        \begin{equation}
            \omega(x, l) = \omega(x, \tilde{l}) = \omega(\pi_{\Lambda_M} x, \tilde{l}),
        \end{equation}
        where the final equality holds because $\tilde{l} \in P_{\Lambda_M}$. As a consequence, the kernel of the map $\pi_{\Lambda_M}: L_{\geq 0}^{\perp} \cap P_{\geq 0} \to P_{\Lambda_M}$ is contained in $(\pi_{\geq 0}L)^{\perp} \cap P_{\geq 0} = L_{\geq 0}$. Thus, $P_{\partial}$ is a~quotient of $\pi_{\Lambda_M}(L_{\geq 0}^{\perp} \cap P_{\geq 0})$. Since $\pi_{\Lambda_M}(L_{\geq 0}^{\perp} \cap P_{\geq 0})$ is an $R_\partial$-submodule of $P_{\Lambda_M}$ and $P_{\Lambda_M}$ is finitely generated over $R_\partial$, $P_{\partial}$ is finitely generated over $R_{\partial}$. 
    \end{proof}


In the definition of the boundary operator module $P_\partial$, one could choose a~different half-space. In general this leads to genuinely different modules (in fact over a different group ring $R_\partial$); for example, $P_\partial$ could be zero for one choice of the vector $v$ in \eqref{eq:half_space}, and nontrivial for another. The simplest example of this is provided by decoupled layers of lower-dimensional systems, for which $P_\partial$ is easy to work out. 

Let us consider the modules of boundary operators for a pair of complementary half-spaces $\Lambda_{\geq 0}, \Lambda_{<0}$:
\begin{subequations}
\begin{align}
{P}_\partial^+ &:= (L_{\geq 0}^\perp\cap P_{\geq 0}) / L_{\geq 0}, \quad \omega_\partial^+=\omega_\partial \text{ descends from } \omega, \\
{P}_\partial^- &:= (L_{<0}^\perp\cap P_{<0}) / L_{<0}, \quad \omega_\partial^- \text{ descends from } \omega.
\end{align}
\end{subequations}
\begin{prop}
The quasi-symplectic modules $P_\partial^+$ and $P_\partial^-$ are opposites in the Witt group of $R_\partial$, see~Section~3.2. In other words, the orthogonal sum
\begin{equation}
(P_\partial^+ \oplus P_\partial^-, \omega_\partial^+ \oplus \omega_\partial^-)
\end{equation}
admits a Lagrangian submodule. In particular, $P_\partial^+$ admits a Lagrangian submodule if and only if $P_\partial^-$ does. 
\end{prop}

\begin{proof}
The map
\begin{equation}
L/(L_{\geq 0} + L_{<0}) \xrightarrow{i} P_\partial^+ \oplus P_\partial^-, \quad
[p] \mapsto \left( [\pi_{\geq 0} p],\ [\pi_{<0} p] \right)
\end{equation}
is injective, and its image $M$ is isotropic: for any $p, q \in L$,
\begin{align}
& (\omega_\partial^+ \oplus \omega_\partial^-)\left( ([\pi_{\geq 0} p], [\pi_{<0} p]), ([\pi_{\geq 0} q], [\pi_{<0} q]) \right) \\
= &\omega(\pi_{\geq 0} p, \pi_{\geq 0} q) + \omega(\pi_{<0} p, \pi_{<0} q)
= \omega(p, q), \nonumber
\end{align}
which vanishes since $L$ is isotropic.

To show $M$ is Lagrangian, suppose $([x], [y]) \in P_\partial^+ \oplus P_\partial^-$, with lifts $x\in L_{\geq 0}^\perp\cap P_{\geq 0}$ and $y\in L_{<0}^\perp\cap P_{<0}$, is in $M^\perp$. Then, for all $p \in L$,
\begin{equation}
\omega(x+y, p)=\omega(x, \pi_{\geq 0} p) + \omega(y, \pi_{<0} p) = 0,
\end{equation}
so $x + y \in L^\perp =L$. Therefore, $([x], [y]) = i(x+y + L_{\geq 0}+L_{<0} )$ belongs to $M$.
\end{proof}

\subsection{Boundary-to-bulk map of charges} \label{sec:bb_map}

        Recall that $E_{P_\partial} := P_\partial^*/P_\partial$ measures the failure of the quasi-symplectic module $(P_\partial, \Omega_\partial)$ to be symplectic. Since $P_\partial$ arises as the boundary operator module of a~stabilizer code $(L, P)$, we denote this $E_{P_\partial}$ simply by $E_L$, and refer to it as the \textbf{boundary charge module}. In general, $E_L$ may be different for different half-spaces used to define $P_\partial$. 
        
        In this subsection we compare the boundary charge module $E_L$ with the bulk charge module $Q_L$. We construct a homomorphism $V : E_L \to Q_L$ and discuss some conditions for $V$ to be injective or surjective. In the case of dimension $d=2$, $V$~is always an isomorphism. In particular, $E_L \cong Q_L$ does not depend on the choice of a half-space for $d=2$ (at least when regarded just as an abelian group rather than a module). We highlight also that this isomorphism allows to draw conclusions about $Q_L$ from the theory of one-dimensional quasi-symplectic modules developed in Section \ref{sec:qs}. The situation is more complicated in higher dimensions.

        

\begin{defn}
    First, define a map $q: L \to P_\partial$ as the composition
    \begin{equation}    
    q: L\xrightarrow{\pi_{\geq0}} \pi_{\geq0}L \hookrightarrow (L_{\geq 0}^\perp\cap P_{\geq0}) \twoheadrightarrow P_\partial.
    \label{eq:qmap_defined}
    \end{equation}
    Explicitly, for $\ell \in L$, $q(\ell)$ is the equivalence class of $\pi_{\geq 0}(\ell)$ in $P_\partial = (L_{\geq 0}^\perp\cap P_{\geq0}) / L_{\geq 0}$. We note that
    \begin{equation}
        \ker(q) = L_{\geq 0}+L_{<0}.
        \label{eq:q_kernel}
    \end{equation}
    
    Let $[\alpha] \in E_L$ be a class represented by 
    \begin{equation}
     \alpha \in P_\partial^* = \Hom_{R_\partial}(\overline{P_\partial}, R_\partial).   
    \end{equation}
    Let $\alpha_0 $ be the constant term of $\alpha$.
    The composition $\alpha_0  q: L \to \mathbb{Z}_n$ is a localized $\mathbb{Z}_n$-valued functional on $L$ (recall the definition around \eqref{eq: star-functional corr}), so it corresponds to an element $\widehat{\alpha q} \in L^*=\mathrm{Hom}_R(\overline L,R)$:
\begin{equation}
    \widehat{\alpha q} (\ell')=\sum_{\lambda \in \mathbb Z^d} (\alpha_0q)(x^\lambda \ell ') x^\lambda = \sum_{n=-\infty}^\infty (\alpha  q)(x_d^n \ell') x_d^n.
    \label{eq:alphaq}
\end{equation}
Indeed, the sum over $n$ in \eqref{eq:alphaq} has only finitely many nonzero terms because 
\begin{equation}
 x^n_d \ell' \in L_{\geq 0} + L_{<0} \subset \ker(q) \qquad \text{for } |n| \gg 0.    
\end{equation}


The map $V : E_L \to Q_L$ is defined by 
\begin{equation}
 V[\alpha] = [\widehat{\alpha q}] \in Q_L = L^*/(P/L).    
\end{equation}
\end{defn}

    \begin{prop}
    The map 
    \begin{equation}
    V: E_L \to Q_L, \qquad V([\alpha]) = [\widehat{\alpha q}]
    \end{equation}
    is a well-defined homomorphism of $R_\partial$-modules.
    \end{prop}
    \begin{proof}
    Let us check that $V$ is well-defined. If $\alpha$ and $\alpha'$ represent the same class in $E_L$, then $\alpha(\cdot) - \alpha'(\cdot) =\omega_\partial(\cdot,p+L_{\geq0})$ for some $p\in L_{\geq 0}^\perp\cap P_{\geq0}$. Hence, 
    \begin{equation}
        (\alpha-\alpha')q(\cdot) = \omega(\cdot,p)
    \end{equation}
    represents the zero class in $Q_L$.

    The map $q$ and the quotient map $L^*\twoheadrightarrow Q_L$ are $R_\partial$-module homomorphisms. The construction $ \mathrm{Hom}_{R_\partial}(\overline L, R_\partial) \ni \alpha q \mapsto \widehat{\alpha q}\in L^*$ is an $R_\partial$-homomorphism as well.  Therefore, $V$ is an $R_\partial$-module homomorphism.
    \end{proof}

    \begin{quest}
        Is $V: E_L \to Q_L$ is injective for every Lagrangian stabilizer code $(L, P)$?
    \end{quest}
    Here is a partial result in this direction.
    \begin{prop} \label{prop:V_inj}
        If there are no secondary boundary operators in the lower half-space, that is, $L_{< 0}^\perp \cap P_{< 0}=\pi_{<0}L$, then $V: E_L \to Q_L$ is injective.
    \end{prop}
    \begin{proof}
        Suppose that $V([\alpha]) = 0$ for some $[\alpha] \in E_L$. Then $\alpha_0 q=\omega(\cdot, p)|_L$ for some $p\in P$. Decompose $p=\pi_{\geq 0}p+\pi_{<0}p$. Notice that if $l\in L_{\geq 0}$, then $q(l)=0$, so
\begin{equation}
0=\alpha_0 q(l)=\omega(l, p)=\omega(l, \pi_{\geq 0}p)+\omega(l,\pi_{<0}p).
\end{equation}
The term $\omega(l,\pi_{<0}p)$ vanishes, and we conclude that $\omega(l, \pi_{\geq 0}p)=0$. Hence $\pi_{\geq 0}p$ is an element of $L_{\geq 0}^\perp \cap P_{\geq 0}$. Similarly, $\pi_{< 0}p\in L_{< 0}^\perp \cap P_{< 0}$. The assumption implies that there exists $x\in L$ such that $\pi_{<0}x=\pi_{< 0}p$. Thus, for any $l\in L$ ,
        \begin{align}
        \alpha_0(q(l))=\omega(l, p) & =\omega(l, \pi_{\geq 0}p)+\omega(l,\pi_{<0}x) \\
        & =\omega(l, \pi_{\geq 0}p-\pi_{\geq 0}x) = \omega(\pi_{\geq 0 } l,x-p) = \omega_{\partial}(q(l), p-x + L_{\geq 0}). \nonumber
        \end{align}
        This shows that $\alpha(\cdot)$ coincides with $\Omega_{\partial}(\cdot, p-x + L_{\geq 0})$ on $P_\partial^{\mathrm{pr}}$, and hence on $P_\partial$.       
    \end{proof}

    In dimensions $d>2$, $V$ may fail to be surjective. Indeed, $Q_L$ regarded as an $R_\partial$-module need not be finitely generated (see Example \ref{exmp:cube model} below); in such cases there can be no surjective homomorphism $E_L \to Q_L$. Below we give a partial characterization of the image of $V$.

    \begin{prop} \label{prop:V_non_surj}
        We have the inclusion
        \begin{equation} 
        \label{eq: trans_mobile}
  V(E_L)\subset \{[h]\in Q_L: \text{there is a lift $h\in L^*$ with } \ker h_0 \supset L_{\geq0}+L_{<0}\}.   
   \end{equation}
   If there are no secondary boundary operators, that is, if $L_{\geq 0}^\perp \cap P_{\geq 0} = \pi_{\geq 0} L$, then the inclusion in \eqref{eq: trans_mobile} is an equality.
    \end{prop}
    \begin{proof}
        The inclusion in \eqref{eq: trans_mobile} is immediate from the definition of $V$ and \eqref{eq:q_kernel}. If there are no secondary boundary operators, the map $q : L \to P_\partial$ defined in \eqref{eq:qmap_defined} is surjective. Let $h \in L^*$. The condition $\ker(h_0) \supset L_{\geq 0} + L_{<0}$ is equivalent to the existence of $\alpha_0 \in \Hom_{\mathbb Z_n}(P_{\partial},\mathbb Z_n)$ such that 
        \begin{equation}
         h_0 = \alpha_0 \circ q.  
         \label{eq:factor_through_q}
        \end{equation}
        Since $h_0$ is localized, $\alpha_0$ is localized and corresponds to an element $\alpha \in P_{\partial}^*$. Clearly $V [\alpha] = [h]$.
    \end{proof}

    \begin{remark}
    In the proof above, absence of secondaries has to be used because otherwise $q$ is not onto $P_{\partial}$, so \eqref{eq:factor_through_q} defines $\alpha_0$ only on $\mathrm{im}(q) = P_{\partial}^{\mathrm{pr}}$. In general, there may exist elements of $(P_\partial^{\mathrm{pr}})^*$ which can't be extended to $P_{\partial}$. Hence \eqref{eq:factor_through_q} can not be used to define $\alpha \in P_{\partial}^*$.   
    \end{remark}


    In the example discussed below, which is often referred to as the X-cube model~\cite{vijay2016fracton}, there exist certain half-spaces for which not all $[h] \in Q_L$ satisfy the condition in~\eqref{eq: trans_mobile}; hence, the map $V$ is not surjective. However, $V$ becomes surjective for the same model with a~differently chosen half-space.

    \begin{exmp} \label{exmp:cube model}
    Let $R=\ZZ_2[x^\pm, y^\pm, z^\pm]$, $P=R^6$ and, let $L$ be the image of the matrix
    \begin{equation}
        \sigma=\begin{pmatrix}
            \sigma_1 & \sigma_2 & \sigma_3
        \end{pmatrix} := \begin{pmatrix}
1+\bar x +\bar y+\bar x \bar y & 0 & 0\\
1+ \bar y + \bar z + \bar y \bar z& 0 & 0\\
1+ \bar x + \bar z + \bar x \bar z& 0 & 0\\
0 & 1+z & 0\\
0 & 1+x & 1+x\\
0 & 0& 1+y
\end{pmatrix}
\end{equation}
Here, $Q_L\cong\Ext^1(P/L, R)$
\begin{equation}
= R/(1+x+y+xy, 1+y+z+yz, 1+z+x+xz)\oplus R^2/\im \tau,
\end{equation}
with 
\begin{equation}
\tau=\begin{pmatrix}
1+z & 1+x & 0\\
0& 1+x& 1+y
\end{pmatrix}.
\end{equation}
Elements of the first summand are represented by $\langle r \rangle\in L^*$, labeled by $r\in R$ and given by 
\begin{equation}
\langle r \rangle(\sigma_1)=r \quad \text{and} \quad \langle r  \rangle(\sigma_2)=\langle r \rangle(\sigma_3)=0.
\end{equation}
We have $\langle r \rangle+\langle s \rangle=\langle r+s \rangle$ for all $r,s\in R$, and $[\langle r \rangle]=0$ in $Q_L$ if and only if $r\in (1+x+y+xy, 1+y+z+yz, 1+z+x+xz)$. 

 Consider the half-space $\Lambda_{\geq 0}=\{(a, b, c)\in \ZZ^3: c\geq 0\}$. We have $x^my^nz\sigma_1\in L_{\geq 0}$ and $\langle x^my^nz\rangle_0(x^my^nz\sigma_1)=1\ne 0$. Moreover, adding to $x^my^nz$ an element of the ideal $(1+x+y+xy, 1+y+z+yz, 1+z+x+xz)$ cannot eliminate monomials of the form $x^my^nz$ without creating another one. Therefore, $[\langle x^my^nz\rangle]\in Q_L$ does not satisfy condition in \eqref{eq: trans_mobile}. Similar counterexamples exist for the half-spaces $\Lambda_{\geq 0}=\{(a, b, c)\in \ZZ^3: a\geq 0\}$ or $\Lambda_{\geq 0}=\{(a, b, c)\in \ZZ^3: b\geq 0\}$. 

On the other hand, for generic slanted half-spaces the condition in~\eqref{eq: trans_mobile} is satisfied for all elements of $Q_L$. An example is
\begin{equation}
\Lambda_{\geq 0}=\{(a, b, c)\in \ZZ^3: a+b+c\geq 0\}.
\end{equation}
One can verify that for this choice of half-space, $V$ is an isomorphism.  
\end{exmp}

The example above highlights a general principle: the obstruction to $V$ being surjective, described in Proposition~\ref{prop:V_non_surj}, arises only for exceptional choices of half-space direction. Below we show that for generic half-spaces $V$ is surjective.

\begin{prop} \label{prop: V_surj_generic}
We say that all local excitations are movable through the boundary if for every class $ [q] \in Q_L$ 
there exist elements $1-f $ and $1-g$ annihilating $[q]$, where $f \in x_d R_{\partial}[x_d]$, $g \in \overline x_d R_{\partial}[\overline x_d]$. In other words, we ask that $f$ involves only monomials $x^{\lambda }$ with $v \cdot \lambda >0 $ and $g$ involves only monomials $x^{\lambda}$ with $v \cdot \lambda <0$, where $v$ is the vector defining the half-space $\Lambda_{\geq 0}$ in \eqref{eq:half_space}.  
\begin{enumerate}
    \item The condition that all local excitations are movable through the boundary is satisfied for generic $v$. More precisely, there exists a finite collection of nonzero vectors $v_1,\dots,v_m \in \mathbb Z^d$ such that if $v \cdot v_i \neq 0$ for each $i$, then the movability condition is satisfied.
    \item If the movability condition is satisfied, then the map $V : E_L \to Q_L$ is surjective.
\end{enumerate}
\end{prop}
\begin{proof}
    1. This argument is based on a similar idea as the proof of Proposition~\ref{prop:topologies}, so~we only sketch it. We can assume $n$ is a prime power. Since $Q_L$ is a~torsion module, we can pick a regular element $r$ in the annihilator of $Q_L$. If $v$ is not orthogonal to any of the vectors
    \begin{equation}
        \{ \lambda - \mu \, | \, \lambda \neq \mu, \ r \text{ has invertible coefficients at } x^\lambda, x^\mu \},
    \end{equation}
    then suitable multiple of $r$ take the form $1-f - f'$ and $1-g-g'$, where $1-f$ and $1-g$ are already as required and $f',g'$ are additional nilpotent terms. Perhaps replacing $1-f$ and $1-g$ by $(1-f)^N$ for some sufficiently large $N$, we can assume that $f'=g'=0$.

    2. Let us take $h \in L^*$. We will show that $[h] = V[\alpha]$ for some $\alpha \in P_\partial^*$. We choose an element of the annihilator of $[h]$ of the form $1-g$ as in the statement of the Proposition. Then $[h] = g^N[h] = [g^N h]$ for every $N \geq 0$, and if $N$ is chosen large enough then $(g^N h)_0$ vanishes on $L_{\geq 0}$. Therefore, with no loss of generality we can assume that $L_{\geq 0} \subset \mathrm{Ker}(h_0)$. 

    Now take $f$ as in the statement of the Proposition. Since $(1-f)[h]=0$, there exists $p \in P$ such that 
    \begin{equation}
        (1-f)h(\cdot) = \left. \Omega(\cdot, p) \right|_L.
    \end{equation}
    We define
    \begin{equation}
        \widetilde p= \sum_{j=0}^\infty f^j p \in \widehat P.
    \end{equation}
    Then $(1-f) \widetilde p = p$, so
    \begin{equation}
        (1-f) ( h (\cdot) - \left. \Omega(\cdot , \widetilde p) \right|_{L} )=0.
        \label{eq:h_ptilde_rep}
    \end{equation}
    This homomorphism is valued in the formal Laurent series ring $R_{\partial}((x_d))$, which consists of power series of the form
    \begin{equation}
        \sum_{j=N}^\infty r_j x_d^j, \qquad \text{for some } N \in \mathbb Z \text{ and } r_j \in R_\partial. 
    \end{equation}
    Such formal power series form a ring $R_\partial((x_d))$, and $1-f$ is a unit in $R_\partial((x_d))$. Hence \eqref{eq:h_ptilde_rep} implies that 
    \begin{equation}
        h (\cdot)  = \left. \Omega(\cdot, \widetilde p ) \right|_L.
    \end{equation}
    In particular, $\omega(l , \widetilde p ) =0$ for $l \in L_{\geq 0}$. We define $\alpha_0 : P_\partial \to \mathbb Z_n$ by
    \begin{equation}
        \alpha_0 (\cdot) = \left. \omega(\cdot, \widetilde p ) \right|_{L_{\geq 0}^\perp \cap P_{\geq 0}}.
    \end{equation}
    This is well-defined on $P_\partial$ because the functional vanishes on $L_{\geq 0}$. It is easy to see that $\alpha_0$ is localized (as a functional on an $R_{\partial}$-module), so $\alpha_0$ is the constant term of certain $\alpha \in P_{\partial}$.

    We will now calculate $V[\alpha]$. To this end, we compute for $l \in L$:
    \begin{equation}
        \alpha_0 \circ q (l) = \omega(\pi_{\geq 0} l, \widetilde p ) = \omega(l, \pi_{\geq 0}\widetilde p) = h_0(l) - \omega(l, \pi_{<0} \widetilde p).
    \end{equation}
    We have $\pi_{<0} \widetilde p \in P$, so the last term corresponds to the zero class in $Q_L$. Hence $V[\alpha]=[h]$.
\end{proof}
\begin{remark}
\begin{itemize}
    \item The sufficient condition on $v$ in 1. of Proposition \ref{prop: V_surj_generic} is likely to be suboptimal. The proof uses the fact that $Q_L$ is annihilated by some regular element of $R$, i.e.\ that the Krull dimension of $Q_L$ is at most $d-1$. It is known that one has a stronger bound $\dim(Q_L) \leq d-2$. 
    \item The existence of $v \in \mathbb Z^d$ not orthogonal to any of $v_1,\dots,v_m \in \mathbb Z^d \setminus \{0 \}$ can be seen as follows. Since $\mathbb Q$ is an infinite field, a finite dimensional $\mathbb Q$-linear space is not the union of finitely many proper subspaces. Therefore, there exists $u \in \mathbb Q^d$ not orthogonal to any $v_i$. We can take $v $ to be a suitable multiple of $u$ belonging to $\mathbb Z^d$. 
\end{itemize}
\end{remark}

Although the above proof is algebraic, it is helpful to visualize $\widetilde{p}$ by plotting the monomials in the Laurent series $\sum_{j=0}^\infty f^j$. Take $1-f=1+x+y+xy$, which is relevant to the X-cube example. Monomials of $\sum_{j=0}^\infty f^j$ are contained in the first quadrant in the $xy$-plane~(region in \textcolor{red}{red} below). Alternatively, we may take $xy(1-f)=xy(1+x^{-1}+y^{-1}+x^{-1}y^{-1})$. Monomials of $\sum_{j=0}^\infty f^j$ are now contained in another quadrant in the $xy$-plane~(region in \textcolor{green}{green} below). In summary, the above proof makes use of power series~(representing extended Pauli operators) contained in cones to formally invert specific polynomials. Moreover, there are multiple formal inverses corresponding to different cones. 

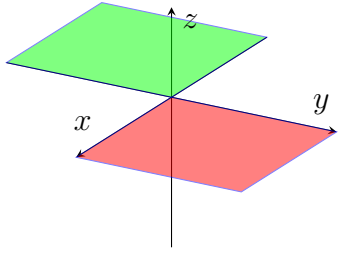
\begin{center}
\begin{tikzpicture}
  \begin{axis}[
    view={120}{30},
    axis lines=center,
    xlabel={$x$},
    ylabel={$y$},
    zlabel={$z$},
    ticks=none,
    xmin=-3, xmax=3,
    ymin=-3, ymax=3,
    zmin=-5, zmax=3,
    ]

    \addplot3[
      surf,
      shader=flat,
      opacity=0.5,
      draw=none,
      fill=red,
      domain=0:3,
      y domain=0:3,
      samples=2,
    ]
    {0};

      \addplot3[
      surf,
      shader=flat,
      opacity=0.5,
      draw=none,
      fill=green,
      domain=0:-3,
      y domain=0:-3,
      samples=2,
    ]
    {0};


  \end{axis}
\end{tikzpicture}
\end{center}

Sometimes, the cone becomes degenerate. In the most degenerate case, cones become rays. This happens particularly for codes satisfying the full mobility condition: $\{1-x_i^l\}_{i=1, \dots, d}$ annihilates $Q_L$ for some $l\ne 0$. This is equivalent to $Q_L$ having Krull dimension zero (or $Q_L=0)$, and also to $Q_L$ being a finite set. 
\begin{cor} \label{cor: surj_mobile}
    If the Krull dimension of $Q_L$ is zero, $V$ is surjective for any choice of half-space~$\Lambda_{\geq 0}$.
\end{cor}
\begin{proof}
    There exists $l \geq 0$ such that $x_i^l-1=-x_i^l(1-x_i^{-l})$ is in the annihilator of $Q_L$ for each $i$. For every half-space there is some $i$ such that one of these polynomials is of the form required in Proposition \ref{prop: V_surj_generic}: if $v=(a_1,\dots,a_d)$, we choose $i$ such that $a_i \neq 0$, and then we can take $f = x_i^l$, $g=x_i^{-l}$. 
%
\end{proof}

Let us reiterate that if we choose the half-space direction as in Proposition \ref{prop: V_surj_generic} and in such a~way that there are no secondaries on $\Lambda_{<0}$, i.e.\ $\pi_{<0} L = L_{<0} \cap P_{<0}$, then by Proposition \ref{prop:V_inj} we have that $V : E_L \to Q_L$ is an isomorphism. We expect that secondary boundary operators are absent for a~generic direction $v$, so that $V$ is generically an isomorphism.

In two dimensions, where \( R = \mathbb{Z}_n[x^{\pm 1}, y^{\pm 1}] \), the analysis simplifies considerably due to mobility results established in~\cite{haah2013commuting, ruba2024homological}. For \textbf{any} choice of half-space \( \Lambda_{\geq 0} \), the boundary-to-bulk map
\[
V: E_L \longrightarrow Q_L
\]
is an isomorphism for any stabilizer code \( (P, L) \), where \( P \) is a standard symplectic \( R \)-module and \( L \subset P \) is Lagrangian. Moreover, both \( E_L \) and \( Q_L \) carry natural \( \mathbb{Z}_n \)-valued quadratic forms—denoted \( q \) and \( \theta \), respectively—arising from the quasi-symplectic structure of \( P_\partial \) in 1D and from the bulk pairing in 2D. These are defined in Section~3.3 (for \( q \)) and Section~6.4 of~\cite{ruba2024homological} (for \( \theta \)). The map \( V \) intertwines these structures, i.e., \( q = \theta \circ V \). Also see Remark~\ref{rmk: additive_q_convention} in Appendix~B for our conventions for quadratic forms. 

\begin{thm}
Let \( R = \mathbb{Z}_n[x^{\pm 1}, y^{\pm 1}] \), and let \( P \cong R^{2q} \) be a standard symplectic module. For any Lagrangian submodule \( L \subset P \) and any half-space \( \Lambda_{\geq 0} \), the boundary-to-bulk map
\[
V: (E_L, q) \longrightarrow (Q_L, \theta)
\]
is an isomorphism of \( R_\partial \)-modules that respects the quadratic forms, i.e., \( q = \theta \circ V \).
\end{thm}
\begin{proof}
     By Corollary~37 from~\cite{ruba2024homological}, the Krull dimension of $Q_L$ is always zero. By Corollary~\ref{cor: surj_mobile}, $V$ is surjective. It is a matter of bookkeeping to verify that that \( q = \theta \circ V \). Since $q$ is non-degenerate, as proven in Section~3.3, this implies that $V$ is injective.
\end{proof}

\begin{cor}
    If $d=2$, the quadratic form $\theta$ on $Q_L$ is non-degenerate, and $Q_L$ admits a Lagrangian subgroup.
\end{cor}

In the case of prime $n$, the result stated in the Corollary above follows from the classification result of Haah \cite{haah2021classification}, which does not cover the significantly richer case of non-prime $n$.



\subsection{Split stabilizer codes} \label{sec:split}

In this subsection, we compare our boundary operator modules~$P_\partial$ with a~construction from the theory of (Clifford) quantum cellular automata (QCAs), which applies when~$L$ is a direct summand of~$P$. We present our bulk-boundary correspondence as a generalization to the non-split case. The QCA bulk-boundary correspondence is, in fact, a case of the fundamental theorem describing the $L$-groups of Laurent extensions~\cite{ranicki1973algebraicII}, so our result may also be viewed as a~conceptual extension of that theorem.

\begin{defn}
A translation-invariant \textbf{Clifford QCA} is a symplectic automorphism of $P$; that is, an $R$-module automorphism that preserves the standard symplectic form $\Omega$. Let $\mathcal{Z} \subseteq P$ denote the submodule consisting of Pauli operators that are products of $Z$ operators only. A stabilizer code $(L, P)$ is said to be \textbf{created} by a Clifford QCA $\alpha$ if $\alpha(\mathcal{Z}) = L$.
\end{defn}

Not all stabilizer codes are created by a Clifford QCA. Conversely, a given stabilizer code may be created by multiple Clifford QCAs. Stabilizer codes that can be produced by a Clifford QCA go by various names depending on the context, including \textbf{invertible stabilizer states}, \textbf{split stabilizer codes}, and \textbf{locally flippable separators}~\cite{haah2023nontrivial}. These names reflect several equivalent characterizations of the same class of codes.
   
    \begin{lem} \label{lem:split_equiv}
    Let $n$ be a prime power, $n=p^t$. The following conditions are equivalent:
    \begin{enumerate}
        \item The short exact sequence $0\rightarrow L \rightarrow P \rightarrow P/L\rightarrow 0$ splits.
        \item $P/L$ is free.
        \item $L$ is free and $Q_L=\Ext^1(P/L, R)=0$. In particular, all topological charges $\Ext^i(P/L, R)$ vanish.   
        \item There exists a Clifford QCA which creates $(L, P)$.
    \end{enumerate}

\end{lem}
    \begin{proof}
        $1 \implies 2:$ $P/L$ is a direct summand of free module $P$, so it is projective. By Quillen-Suslin theorem, every finitely generated projective module over $R/pR$ is free. This statement can be extended to $R$-modules as follows. If $M$ is a finitely generated projective $R$-module, then $M/pM$ is projective and hence free over $R/pR$. Lift a~basis of $M/pM$ to $\{m_1, \dots, m_k\}\subset M$. Then $m_1,\dots,m_k$ generate $M$ by Nakayama's lemma. We obtain a short exact sequence 
        \begin{equation}
         0\rightarrow K\rightarrow R^k\xrightarrow{\sigma} M\rightarrow 0,   
        \end{equation}
        where $\sigma(r_i)_{i=1}^k = \sum_i r_i m_i$ and $K=\ker \sigma$. This sequence splits because $M$ is projective; hence it remains exact after reduction mod $p$. On the other hand, $\sigma$~becomes an isomorphism after tensoring with $R/pR$, so $K/pK=0$. Hence, by~Nakayama's lemma again, $K=0$. 
        
        $2 \implies 3:$ Since $P/L$ is free, the short exact sequence splits, so $L$ is projective and hence free as well. The vanishing of $Q_L$ follows from the definition of $\Ext$. 
        
        $3 \implies 1:$ 
        As $L$ is free, $\Ext^1(P/L, R)=0$ implies $\Ext^1(P/L, L)=0$ as well. In other words, all extensions of $P/L$ by $L$ are trivial and there exists $P\cong L\oplus P/L$. 
        
        $4 \implies 1$ is obvious. To see that $1 \implies 4$, one uses the standard fact that the splitting $P/L \to P$ can be chosen so that its image is a Lagrangian submodule~$L'$. Then $L,L'$ are free of rank half the rank of $P$ and canonically dual via $\Omega$, $L' \cong L^*$. One can define $\alpha$ by sending $\mathcal Z$ to $L$ by an arbitrarily chosen isomorphism, and the module $\mathcal X$ of $X$ operators to $L'$ via the inverse dual isomorphism. 
        
        For the convenience of the reader we briefly review why the splitting can be chosen Lagrangian. The map
        \begin{equation}
            P/L \ni p + L \mapsto \left. \Omega (\cdot , p) \right|_{L} \in L^* 
        \end{equation}
        is an isomorphism, and we use it as identification. We choose a splitting $s : L^* \to P$ of the quotient map $P \to L^*$ and let 
        \begin{equation}
            D(\varphi,\psi) := \Omega(s(\varphi),s(\psi)).
        \end{equation}
        We need to adjust $s$ so that $D=0$. Let $\varphi_1,\dots,\varphi_t$ be a basis of $L^*$. We consider the matrix $D_{ij} = D(\varphi_i,\varphi_j)$. Then $D_{ij} = - \overline{D_{ji}}$, and $D_{ii}$ has zero constant term. Hence there exists a matrix $H_{ij}$ such that $D_{ij} = H_{ij} - \overline{H_{ji}}$. We let $H$ be the sesquilinear form on $L^*$ characterized by $H(\varphi_i,\varphi_j)= H_{ij}$. We have
        \begin{equation}
            D(\varphi,\psi) = H(\varphi,\psi) - \overline{H(\psi,\varphi)}.
        \end{equation}
        There exists $h : L^* \to L$ such that
        \begin{equation}
            H(\varphi,\psi) = \Omega(s(\varphi),h(\psi)).
        \end{equation}
        The splitting $s-h$ is as required.
    \end{proof}



    
\begin{remark}
    Split stabilizer codes provide examples of invertible gapped local Hamiltonians. Classification of the latter is part of a~conjecture by Kitaev. On the other hand, split stabilizer codes are simply classified by Clifford QCAs that create them, up to appropriate equivalence relation. This problem have been resolved in~\cite{haah2025topological, yang2025categorifying} using algebraic $L$-theory. The latter employs bounded $L$-theory~\cite{ranicki1992lower} to obtain a result that remains valid under coarse equivalence of the underlying geometric space. In particular, the assumption of translation invariance on Euclidean lattices is no longer required. It is conceivable that an analogous generalization should also be possible for the results of the present paper.

\end{remark}

Clifford QCAs belong to the actively studied topic of quantum cellular automata (QCAs). A bulk-boundary correspondence has been established for a~QCA and its boundary algebra~\cite{haah2021clifford, haah2023invertible}. 
We will see that the boundary algebra of $\alpha$ (as defined below) is Witt equivalent to the code’s boundary operator module $P_{\partial}$. Therefore, $P_{\partial}$ provides a natural extension of the QCA boundary algebra concept to non-split stabilizer codes, even in the absence of a generating QCA. 

We introduce $R_\partial$-submodules of operators supported in slabs:
{\begin{equation}
      P_{[m,n]} \;=\; P_{\Lambda_{m,n}}, 
  \quad
  \Lambda_{m,n} 
    := \{(\mathbf{m},h)\in\ZZ^{d-1}\times\ZZ : m\le h\le n\},
  \quad m \leq n.
\end{equation}

Also, set
\begin{equation}
  P_{\le r} \;=\; P_{\Lambda_{\le r}}, 
  \quad
  \Lambda_{\le r} 
    := \{a=(a_1,\dots,a_d)\in\ZZ^d : a_d\le r\},
\end{equation} and similarly for $P_{> r}$.
By definition, the \textbf{spread} $l$ of a Clifford QCA $\alpha$ is the smallest nonnegative integer such that
\begin{equation}
    \alpha\bigl(P_{[m,n]}\bigr)\subset P_{[m-l,\;n+l]}
  \quad\text{for all }m<n.
\end{equation}
Its inverse $\alpha^{-1}$ has the same spread. Indeed, since $\alpha$ is a symplectic automorphism, for~any $p \in P_{[m,n]}$, we have that $\alpha^{-1}(p)$ is orthogonal to any $x \in P_{[m-l, n+l]}^\perp$. Hence, $\alpha^{-1}(p) \in P_{[m-l, n+l]}$.

For any $r\ge l$,~\cite{haah2023invertible} defines the \textbf{boundary algebra} as
\begin{equation}
      B^r(\alpha)
    := \alpha(P_{\le r}) \;\cap\; P_{[0,\;r+l]}
    = \alpha(P_{\le r}) \;\cap\; P_{\ge0}.
\end{equation}

Notice $P_{[0,\;r+l]}$ is a standard symplectic $R_\partial$-module and $B^r(\alpha)$ is its $R_\partial$-submodule. It is equipped with a skew-hermitian form given by the restriction of the one on $P_{[0,\;r+l]}$. This construction originally appears in L-theory~\cite{ranicki1973algebraicII} where it was proven that $B^r(\alpha)$ is projective and symplectic\footnote{In that context, this property is more commonly referred to as unimodularity.}. Moreover, its Witt class is independent of $r\geq l$. We provide an argument below.

\begin{lem}
    Let
\begin{equation}
    D^r(\alpha) := \alpha(P_{> r}) \cap P_{[0,\;r+l]}.
\end{equation}
Then we have the orthogonal direct sum decomposition
\begin{equation}
    D^r(\alpha) \oplus B^r(\alpha) = P_{[0,\;r+l]}.
\end{equation}
\end{lem}
\begin{proof}
The intersection $D^r(\alpha) \cap B^r(\alpha) = \{0\}$ is trivial because
\begin{equation}
    \alpha(P_{\le r}) \cap \alpha(P_{> r}) = \{0\}.
\end{equation}
For similar reasons, $D^r(\alpha)$ and $B^r(\alpha)$ are orthogonal.

To show that $D^r(\alpha) + B^r(\alpha) = P_{[0,\;r+l]}$, take \( p \in P_{[0,\;r+l]} \). Write
\begin{equation}
    \alpha^{-1}(p) = a^- + a^+, \quad \text{with } a^- \in P_{\le r},\; a^+ \in P_{> r}.
\end{equation}
Then
\begin{equation}
    p = \alpha(a^-) + \alpha(a^+),
\end{equation}
with
\begin{equation}
    \alpha(a^-) \in P_{\le r+l}, \quad \alpha(a^+) \in P_{> r-l} \subset P_{>0}.
\end{equation}
Since \( p \in P_{[0,\;r+l]} \), it follows that
\begin{equation}
    \alpha(a^-),\; \alpha(a^+) \in P_{[0,\;r+l]}.
\end{equation}
Therefore,
\begin{equation}
    \alpha(a^-) \in B^r(\alpha), \quad \alpha(a^+) \in D^r(\alpha).
\end{equation}
\end{proof}

Lemma above implies that $B^r(\alpha)$ is a projective module, hence a free module (see the proof of Lemma \ref{lem:split_equiv}), and that it is symplectic: the standard symplectic form on \( P_{[0,\;r+l]} \) restricts to non-degenerate (and Witt-opposite) forms on \( B^r(\alpha) \) and \( D^r(\alpha) \).

Now, for any \( \delta > 0 \), translation symmetry gives an isomorphism
\begin{equation}
    B^{r+\delta}(\alpha)
    = \alpha(P_{\le r+\delta}) \cap P_{[0,\;r+\delta+l]} 
    \cong \alpha(P_{\le r}) \cap P_{[-\delta,\;r+l]}.
\end{equation}
If \( r \geq l \), the latter module decomposes as
\begin{equation}
    B^r(\alpha) \oplus P_{[-\delta,\;0]}.
\end{equation}
Since \( P_{[-\delta,\;0]} \) is a standard symplectic module, the Witt class of \( B^r(\alpha) \) is independent of \( r \ge l \).

Since $P_\partial$ is finitely generated, there exists $M > 0$ such that
\begin{equation} \label{eq: strip_support}
    P_{[0, M]} + L_{\geq 0} \supset L_{\geq 0}^\perp \cap P_{\geq 0}.
\end{equation}
Since $\alpha(\mathcal{Z}) = L$, one can take\footnote{The term $P_{[0,M]}$ in \eqref{eq: strip_support} is not needed at all if $l=0$, in which case $L_{\geq 0} = L_{\geq 0}^\perp \cap P_{\geq 0}$.} $M = 2l-1$. We keep $M$ and $l$ independen to clarify the roles each quantity plays in the proofs that follow. Let us show that one can take $M =2l-1$. We have the inclusion $\alpha(\mathcal Z_{\geq l}) \subset L_{\geq 0}$, so
\begin{equation}
   L_{\geq 0}^\perp \subset \alpha( \mathcal Z_{\geq l}^\perp) = \alpha( \mathcal Z_{\geq l} + P_{<l}) = \alpha( \mathcal Z_{\geq l}) + \alpha(P_{<l}) \subset L_{\geq 0} + P_{<2l}. 
\end{equation}
Hence $L_{\geq 0}^\perp \subset L_{\geq 0} + P_{[0,2l-1]}$. 

\begin{lem} \label{lem: strip_boundary}
Let $B := B^{M+l}(\alpha)$. Then
\begin{equation}
    P_{[0, M]} \subset B.
\end{equation}
\end{lem}

\begin{proof}
For any $p \in P_{[0, M]}$, we have $\alpha^{-1}(p) \in P_{[-l, M+l]} \subset P_{\leq M+l}$.
\end{proof}

\begin{lem} \label{lem: boundary_well-defined}
The isotropic $R_\partial$-submodule $L_B := L \cap B$ of $B$ satisfies
\begin{equation}
    L_B^\perp \cap B \subset L_{\geq 0}^\perp \cap P_{\geq 0}.
\end{equation}
\end{lem}

\begin{proof}
Throughout the proof, let $r := l + M$.

\smallskip
\noindent
Let $b \in L_B^\perp \cap B$. Then $b \in B \subset P_{\geq 0}$, and there exists $c \in P_{\leq r}$ such that
\begin{equation}
    \alpha(c) = b.
\end{equation}
Given $x \in L_{\geq 0}$, we have $z := \alpha^{-1}(x) \in \mathcal{Z}$. Since $\alpha$ preserves the symplectic form $\omega$, and the projection $\pi_{\leq r}$ onto $P_{\leq r}$ satisfies Lemma~\ref{lem: orthogonal}, we compute:
\begin{equation}
    \omega(b, x) = \omega(c, z) = \omega(c, \pi_{\leq r} z) = \omega(b, \alpha(\pi_{\leq r} z)).
\end{equation}
This vanishes if $\alpha(\pi_{\leq r} z) \in L_B$. First, note that $\alpha(\pi_{\leq r} z) \in L$, since $\alpha(\mathcal{Z}) = L$ and $\pi_{\leq r} z \in \mathcal{Z}$. Next, since $\pi_{\leq r} z \in P_{\leq r}$, it remains to show that $\alpha(\pi_{\leq r} z) \in P_{\geq 0}$. Observe:
\begin{equation}
    \alpha(\pi_{\leq r} z) = x - \alpha(\pi_{>r} z) ,
\end{equation}
where $\alpha(\pi_{> r} z) \in P_{>r-l} \subset P_{\geq 0}$.
\end{proof}

\begin{prop}\label{prop: Bounary_boundary}
     There is an isomorphism 
     \begin{align}
         \frac{L_B^\perp\cap B}{L_B}&\xrightarrow[]{\sim} \frac{L_{\geq 0}^\perp\cap P_{\geq 0}}{L_{\geq 0}}=P_\partial \quad \text{given by}\\
         b+L_B &\mapsto b + L_{\geq 0}. \nonumber
     \end{align}
\end{prop}
\begin{proof}
The map is well-defined due to Lemma~\ref{lem: boundary_well-defined} and the inclusion $L_{B} \subset L_{\geq 0}$. It~is surjective because of inclusion~\eqref{eq: strip_support} and Lemma~\ref{lem: strip_boundary}. Finally, it is injective: if~$b\in L_B^\perp \cap B$ belongs to $L_{\geq 0}$ then it belongs to $L_B=L\cap B.$
\end{proof}
\begin{cor}
    The boundary algebra $B$ and the module of boundary operators $P_\partial$ are equivalent in the Witt group $W(R_\partial)$. 
\end{cor}
\begin{proof}
By Lemma~\ref{lem:reduction_Witt_equivalent}, $B$ is Witt equivalent to $(L_B^\perp \cap B)/L_B$.
\end{proof}

    \appendix

    \section{Groups of Heisenberg type} \label{app:Heisenberg}

    \subsection{Central extensions}
    \label{app:central}

    Let $A, Z$ be abelian groups. In this Appendix we discuss central extensions
    of $A$ by $Z$, that is, groups $H$ containing $Z$ as a central subgroup and with
    a fixed isomorphism $H / Z \cong A$. We will use additive notation for the
    group operation in $A$, and multiplicative notation for those in $H$ and $Z$.

    It is well known that central extensions of $A$ by $Z$ are classified
    up to isomorphism of extensions (i.e.\ isomorphisms respecting the inclusions
    and quotient identifications) by the second cohomology group $H^{2}(A,Z)$. In
    this appendix, we briefly review this classification and examine the extent
    to which a~central extension can be determined by its commutator form.

    Given a central extension $H$, the associated commutator form is the bilinear
    map $\omega : A \times A \to Z$ defined by
    \begin{equation}
        T_{a} T_{a'}= T_{a'}T_{a} \omega(a,a') \label{eq:commutator_form}
    \end{equation}
    for elements $T_{a},T_{a'}\in H$ representing $a,a' \in A$. This form
    $\omega$ is well-defined: it does not depend on the choice of representatives.
    It is alternating, meaning $\omega(a,a)=1$ for all $a \in A$, and hence in
    particular skew-symmetric.

    We will see that if $Z$ is a \textit{divisible group}, commutator forms
    classify central extensions of $A$ by $Z$. In this case the cohomology group
    $H^{2}(A,Z)$ can be identified with the group of all alternating bilinear forms
    $A \times A \to Z$. The case we of primary interest to us is
    $Z = \mathbb{T}$, the circle group, which is a divisible group.

    Let $H$ be a central extension of $A$ by $Z$. For~every $a \in A$ we fix a
    representative $T_{a} \in H$. Then every element of $H$ can be uniquely written
    in the form $z T_{a}$, where $z \in Z$ and $a \in A$. The multiplication in $H$
    takes the form
    \begin{equation}
        z T_{a} \cdot z' T_{a'}= zz' \chi(a,a') T_{a+a'}, \label{eq:extension_mult}
    \end{equation}
    where $\chi$ is a function $A \times A \to Z$. Associativity of the group
    operation in $H$ imposes that $\chi$ satisfies the $2$-cocycle condition:
    \begin{equation}
        \chi(a,a')\chi(a+a',a'') = \chi(a,a'+a'') \chi(a',a'').
    \end{equation}
    Conversely, given a $2$-cocycle $\chi$ one can construct a central extension
    $H$ as the Cartesian product $Z \times A$ with multiplication defined by
    \begin{equation}
        (z,a) \cdot (z',a') = (zz' \chi(a,a'), a+a').
    \end{equation}
    With the identification $T_{a} = (1,a)$, this is simply a concrete realization
    of the multiplication rule in \eqref{eq:extension_mult}.

    Now let $\widetilde H$ be another central extension of $A$ by $Z$. For every
    $a \in A$ we choose a~representative $\widetilde T_{a} \in \widetilde{H}$.
    The multiplication in $\widetilde{H}$ is described by a $2$-cocycle
    $\widetilde \chi$, via~a~formula analogous to \eqref{eq:extension_mult}. A homomorphism
    of central extensions is a~group homomorphism $H \to \widetilde{H}$ of the form
    \begin{equation}
        z T_{a} \mapsto z \phi(a) \widetilde T_{a}, \label{eq:extension_homomorphism}
    \end{equation}
    where $\phi : A \to Z$ is a function. Every such homomorphism is clearly
    bijective. In~order for \eqref{eq:extension_homomorphism} to define group
    homomorphism, the function $\phi$ must satisfy
    \begin{equation}
        \widetilde \chi(a,a') = \chi(a,a') \phi(a+a') \phi(a)^{-1}\phi(a')^{-1}.
        \label{eq:cohomologous}
    \end{equation}

    Two $2$-cocycles $\chi,\widetilde \chi$ are called cohomologous if there
    exists a function $\phi : A \to Z$ such that \eqref{eq:cohomologous} holds.
    The set of equivalence classes of $2$-cocycles under this relation forms the
    second cohomology group $H^{2}(A,Z)$. As discussed above, the elements of
    $H^{2}(A,Z)$ classify central extensions of $A$ by $Z$ up to isomorphism.

    Given any extension $H$ of $A$ by $Z$ there is an associated commutator form
    defined as in \eqref{eq:commutator_form}. Using the $2$-cocycle $\chi$ that
    describes the extension, this commutator form $\omega$ is given by
    \begin{equation}
        \omega(a,a') = \chi(a,a') \chi(a',a)^{-1}. \label{eq:omega_from_chi}
    \end{equation}
    Cohomologous $2$-cocycles give rise to the same commutator form. Therefore, the~above
    formula defines a group homomorphism
    \begin{equation}
        H^{2}(A,Z) \ni \text{class of }\chi \mapsto \omega \in \{ \text{alternating
        bilinear forms }A \times A \to Z \}. \label{eq:cohomology_to_alternating}
    \end{equation}

    \begin{prop}
        \label{prop:cohomology_commutator_forms}
        \begin{enumerate}
            \item The map \eqref{eq:cohomology_to_alternating} is surjective. In
                other words, every alternating bilinear form $A \times A \to Z$
                is the commutator form of some extension of $A$ by~$Z$.

            \item If $Z$ is divisible, the map \eqref{eq:cohomology_to_alternating}
                is injective. Thus extensions of $A$ by $Z$ are determined up to
                isomorphism by their commutator form.
        \end{enumerate}
    \end{prop}
    \begin{proof}

        1. We choose a generating set $\{ a_{i} \}_{i \in I}$ of $A$ and a total
        ordering of $I$. We~have an epimorphism
        \begin{equation}
            \mathbb{Z}^{\oplus I}\ni (n_{i})_{i \in I}\mapsto \sum_{i \in I}n_{i}
            a_{i} \in A.
        \end{equation}
        We denote the kernel of this map by $K$. We define a group
        \begin{equation}
            H^{\mathrm{pre}}= Z \times \mathbb{Z}^{\oplus I}
        \end{equation}
        with the multiplication
        \begin{equation}
            (z, (n_{i})_{i \in I}) \cdot ( z', (n_{i}')_{i \in I}) = \left(zz' \prod
            _{i >j}\omega(a_{i},a_{j})^{n_i n_j'}, (n_{i} + n_{i}')_{i \in I}\right
            ).
        \end{equation}
        Then $H^{\mathrm{pre}}$ is a central extension of $\mathbb{Z}^{\oplus I}$
        by $Z$ with the commutator form
        \begin{equation}
            ((n_{i})_{i \in I}, (n_{i}')_{i \in I}) \mapsto \omega \left( \sum_{i
            \in I}n_{i} a_{i}, \sum_{j \in J}n_{j} a_{j} \right),
        \end{equation}
        which lifts $\omega$ from $A$ to $\mathbb{Z}^{\oplus I}$. Now let
        \begin{equation}
            \widetilde K = \{ (1, (n_{i})_{i \in I}) \in H^{\mathrm{pre}}\, | \,
            (n_{i})_{i \in I}\in K \}.
        \end{equation}
        Then $\widetilde K$ is a central subgroup of $H^{\mathrm{pre}}$
        isomorphic to $K$ and meeting $Z$ trivially. The quotient $H = H^{\mathrm{pre}}
        /\widetilde K$ is a central extension of $A$ by $Z$ with commutator form
        $\omega$.

        2. If $\chi$ is in the kernel of the map \eqref{eq:cohomology_to_alternating},
        then the corresponding extension $H$ is an abelian group. Since $Z$ is
        divisible, it splits $H$: $H \cong Z \oplus A$. Hence $\chi$ is cohomologous
        to the trivial cocycle ($\chi(a,a') = 1$ for all $a,a' \in A$).
    \end{proof}

    The following example illustrates that the assumption that $Z$ is divisible can't
    be removed in point 2. of Proposition \ref{prop:cohomology_commutator_forms}.

    \begin{exmp}
        Let $A = \mathbb{Z}_{2} \times \mathbb{Z}_{2}$. Consider the alternating
        form $\omega : A \times A \to \mathbb{Z}_{2}$ given~by
        \begin{equation}
            \omega((a,b),(a',b')) = ab'-a'b.
        \end{equation}
        There exist two non-isomorphic central extensions of $A$ by
        $\mathbb{Z}_{2}$ with commutator form $\omega$: the dihedral group of order
        $8$ and the quaternion group. In terms of linear representations, these groups
        are generated respectively by the pairs $\sigma_{x},\sigma_{z}$ and $i \sigma
        _{x}, i \sigma_{z}$, where $\sigma_{x},\sigma_{z}$ are the Pauli
        matrices:
        \begin{equation}
            \sigma_{x} =
            \begin{bmatrix}
                0 & 1 \\
                1 & 0
            \end{bmatrix}, \qquad \sigma_{z} =
            \begin{bmatrix}
                1 & 0  \\
                0 & -1
            \end{bmatrix}.
        \end{equation}
        When viewed as groups of operators, the distinction disappears upon
        adjoining all complex phase factors\footnote{In this case, adjoining just
        the imaginary unit suffices, yielding the so-called Pauli group of order
        $16$.}.
    \end{exmp}

    \begin{cor} \label{cor:central_ext_homomorph}
        Let $Z$ be a divisible group and let $A,A'$ be abelian groups with $Z$-valued
        alternating forms $\omega,\omega'$. Let $H,H'$ be the corresponding
        central extensions of $A,A'$ by $Z$. Suppose that $f : A \to A'$ is a homomorphism
        such that
        \begin{equation}
            \omega'(f(a_{1}),f(a_{2}))=\omega(a_{1},a_{2}) \qquad \text{for all }
            a_{1}, a_{2} \in A. \label{eq:symplectic_map}
        \end{equation}
        Then $f$ lifts to a homomorphism $F: H \to H'$ satisfying $F(z) =z$ for $z
        \in Z$.
    \end{cor}
    \begin{proof}
        If $a \in A$ and $a' \in A'$, we choose representatives $T_{a} \in H$
        and $T_{a'}' \in H$. We let $\chi,\chi'$ be the $2$-cocycle characterizing
        $H$ and $H'$, respectively. The sought homomorphism $F$ has to be of the
        form
        \begin{equation}
            F(z T_{a}) = z \phi(a) T_{f(a)}' \label{eq:symplectic_map_lift}
        \end{equation}
        for some function $\phi : A \to Z$. In order for \eqref{eq:symplectic_map_lift}
        to be a homomorphism, $\phi$ has to satisfy
        \begin{equation}
            \phi(a_{1}+a_{2}) \phi(a_{1})^{-1}\phi(a_{2})^{-1}= \chi'(f(a_{1}),f(
            a_{2})) \chi(a_{1},a_{2})^{-1}. \label{eq:mapping_cochain_exact}
        \end{equation}
        The right hand side of this equation is a $2$-cocycle on $A$ whose
        associated commutator form vanishes by the assumption \eqref{eq:symplectic_map}.
        By Proposition \ref{prop:cohomology_commutator_forms} such cocycles are
        cohomologous to the trivial one, i.e.\ a function $\phi$ satisfying
        \eqref{eq:mapping_cochain_exact} exists.
    \end{proof}

    \subsection{Clock and shift systems}

    Let $\mathbb{C} \mathbb{Z}_{n}$ be a vector space with basis
    $(e_{k})_{k \in \mathbb{Z}_n}$. We define operators $X$ and $Z$ on $\mathbb{C}
    \mathbb{Z}_{n}$ by
    \begin{equation}
        Ze_{k} = e^{\frac{2 \pi i k}{n}}e_{k}, \qquad X e_{k} = e_{k+1}.
    \end{equation}
    These operators are sometimes called the \textbf{clock} and \textbf{shift} operators
    (the eigenvalue of $Z$ can be thought of as the hand of the clock, and $X$ shifts
    it forward). They satisfy the relations
    \begin{equation}
        X^{n}=Z^{n}=1, \qquad ZX=XZ e^{\frac{2\pi i}{n}}.
    \end{equation}
    The group generated by $Z,X$ and all complex phase factors is a~central extension
    of $\mathbb{Z}_{n} \oplus \mathbb{Z}_{n}$ by the circle group. When $n=2$, the
    operators $X$ and $Z$ coincide with the Pauli matrices $\sigma_{x}$ and $\sigma
    _{z}$. For this reason, the group they generate--perhaps together with some
    complex phase factors--is called the Pauli group.

    The tensor product $\bigotimes_{j=1}^{t} \mathbb{C} \mathbb{Z}_{n_j}$ carries
    an action of $t$ pairs of clock and shift operators $Z_{j},X_{j}$, each
    acting on the corresponding tensors factors. The group $H$ generated by all these
    operators is a central extension of $A = \bigoplus_{j=1}^{t} (\mathbb{Z}_{n_j}
    \oplus \mathbb{Z}_{n_j})$ by the circle group. The~commutator form
    associated to this extension is explicitly given by
    \begin{equation}
        \omega \left( (a_{j},b_{j})_{j=1}^{t}, (a_{j}',b_{j}')_{j=1}^{t} \right)
        = \exp \left( 2 \pi i \sum_{j=1}^{t} \frac{a_{j} b_{j}' -a_{j}' b_{j}}{n}
        \right). \label{eq:canonical_omega}
    \end{equation}

    It turns out that up to isomorphisms, every central extension by
    $\mathbb{T}$ of a finite abelian group described by a nondegenerate
    commutator form is isomorphic to the group generated by some number of pairs
    of clock and shift operators. This fact is formalized in the statement below.

    \begin{prop}
        Let $A$ be a finite abelian group and let $\omega : A \times A \to \mathbb{T}$
        be an alternating form on $A$ which is nondegenerate, i.e.\ if $a \in A$
        is such that $\omega(a,a')=1$ for all $a' \in A$, then $a=0$.
        \begin{enumerate}
            \item There exists a group isomorphism
                \begin{equation}
                    \phi \, : \, \bigoplus_{j=1}^{t} (\mathbb{Z}_{n_j}\oplus \mathbb{Z}
                    _{n_j}) \to A
                \end{equation}
                such that $\omega \circ (\phi \times \phi)$ is given by \eqref{eq:canonical_omega}.

            \item Identify $A$ with $\bigoplus_{j=1}^{t} (\mathbb{Z}_{n_j}\oplus
                \mathbb{Z}_{n_j})$ and let $H$ be a central extension by $\mathbb{T}$
                with commutator form $\omega$. Up to isomorphism, the only irreducible
                representation of $H$ on which $\mathbb{T}$ acts tautologically (i.e.
                where each complex phase $z \in \mathbb{T}$ acts as scalar multiplication
                by $z$) is the representation
                $\bigotimes_{j=1}^{t} \mathbb{C} \mathbb{Z}_{n}$.
        \end{enumerate}
    \end{prop}
    \begin{proof}
        1. The proof of this assertion is elementary and mimics how one brings symplectic forms on
        vector spaces to a canonical form.

        2. $\mathcal{P}$ contains a subgroup $\mathcal{P}_{0}$ which is a central
        extension of $P$ by $\mathbb{Z}_{n}$, where $n$ is the least common
        multiple of $n_{1},\dots,n_{t}$. The problem reduces to studying representations
        of the finite group $\mathcal{P}_{0}$, which can be done with elementary
        character theory. We omit the details.
    \end{proof}

    \subsection{Stabilizer states}
    \label{app:stabilizer}

    In the final part of this Appendix, we construct stabilizer states and discuss their purity and uniqueness. 

    Let $A$ be an abelian group and let $\omega : A \times A \to \mathbb T$ be an alternating bilinear form valued in the circle group. We assume that $\omega$ is not identically equal $1$ to avoid the trivial case. A subgroup $L \subset A$ is said to be isotropic if $\omega(l,l')=1$ for all $l,l' \in L$, and Lagrangian if $\left. \omega(a,\cdot) \right|_L =1$ holds if and only if $a \in L$. The nontriviality assumption on $\omega$ guarantees that isotropic subgroups are proper (not all of $A$).
    
    Let $H$ be the central extension of $A$ by $\mathbb T$ with commutator form~$\omega$. We define $\mathcal A$ to be the quotient of the complex group ring $\mathbb C[H]$, obtained by identifying the circle subgroup of $H$ with the circle subgroup of the ground field~$\mathbb C$. It is an associative $\star$-algebra, with the involution extending anti-linearly the inverse operation in $H$.

    If $L \subset A$ is an isotropic subgroup, then, by Corollary \ref{cor:central_ext_homomorph}, the quotient map $H \to A$ admits a section on $L$. That is, we can choose a representative $T_l \in H$ for every $l \in L$ in such a way that the map $l \mapsto T_l$ is a group homomorphism. We can extend this map and pick a representative $T_a \in H$ for every $a \in A$, but the map $a \mapsto T_a$ is not a group homomorphism on the whole $A$.
    
    We denote the subgroup of $H$ consisting of elements $T_l$ by $\mathcal L$. We have $\mathcal L \cong L$ as abstract groups, but as a subgroup of $H$, $\mathcal L$ depends on the choice made above; it is not uniquely determined by $L$ itself. The group $\mathcal L$ intersects $\mathbb T$ trivially.

    Let $\alpha$ be any group homomorphism $\mathcal L \to \mathbb T$. We would like to find a representation of $H$ containing a joint eigenvector of all $T_l$ to eigenvalues $\alpha(l)$:
    \begin{equation}
        T_l v = \alpha(l) v.
        \label{eq:stabilizer_conditions}
    \end{equation}
    We will call such vector $v$ a stabilizer state vector. 
    
    In order to construct stabilizer state vectors, let $I_{\mathcal L, \alpha}$ be the left ideal of $\mathcal A$ generated by elements $T_l - \alpha(l)$. The quotient $V_{\mathcal L, \alpha} = \mathcal A / I_{\mathcal L,\alpha}$ is an $\mathcal A$-module with a cyclic vector satisfying \eqref{eq:stabilizer_conditions}. One can also construct an $H$-invariant sesquilinear form on $V_{\mathcal L, \alpha}$ and complete this space to obtain a unitary representation of $H$ on a~Hilbert space; we omit the details. We remark that the Hilbert space completion of $V_{\mathcal L,\alpha}$ can be identified with the GNS representation for the state (positive and unital linear functional) on $\mathcal A$ given by
    \begin{equation}
        T_a \mapsto \begin{cases}
            \alpha(a), & a \in L, \\
            0, & \text{otherwise}.
        \end{cases}
        \label{eq:GNS_state}
    \end{equation}
    In this language, the next Proposition discusses the purity of the state \eqref{eq:GNS_state}.

    \begin{prop} \label{prop:stab_rep_irreducible}
    $V_{\mathcal L,\alpha}$ is an irreducible representation (or a simple module) of $\mathcal A$ if and only if $L \subset A$ is a~Lagrangian subgroup.
    \end{prop}
    \begin{proof}
    It is clear that every element of $V_{\mathcal L,\alpha}$ admits a representative in $\mathcal A$ of the form 
    \begin{equation}
        v = \sum_{i=1}^n c_i T_{a_i},
        \label{eq:vector_in_rep}
    \end{equation}
    where $a_1,\dots,a_n \in A$ represent distinct classes in $A/L$ and all $c_i$ are nonzero (the case of the zero class in $V_{\mathcal L, \alpha}$ is covered by the possibility $n=0$). We claim that if $n \neq 0$, the vector \eqref{eq:vector_in_rep} represents a nonzero element in $V_{\mathcal L, \alpha}$. Indeed, let $\varepsilon : \mathcal A \to \mathbb C$ be a linear map such that
    \begin{equation}
        \varepsilon(T_{a_1} T_l) = \alpha(l), \qquad \varepsilon (T_b)=0 \text{ if } b-a_1 \not \in L.
    \end{equation}
    Then $\varepsilon(v) = c_1 \neq 0$. On the other hand, it is easy to check that $I_{\mathcal L, \alpha} \subset \ker(\varepsilon)$. Hence $v \not \in I_{\mathcal L,\alpha}$.

    The above argument establishes that $V_{\mathcal L,\alpha}$ has a basis given by classes of elements $T_{a_i}$, where $\{ a_i \}_{i \in I}$ are representatives, one for each element of $A/L$.

    Suppose that $L$ is not Lagragian. Pick an element $l' \in A \setminus L$ which is orthogonal to $L$, $\omega(l,l')=1$ for all $l \in L$. We also have $\omega(l',l')=1$ because $\omega$, so the subgroup $L'$ generated by $L$ and $l'$ is isotropic and properly contains $L$. We lift $L'$ to a subgroup $\mathcal L'$ as we have for $L$ and $\mathcal L$. Since $\mathbb T$ is divisible, $\alpha$ extends to a homomorphism $\alpha' : \mathcal L' \to \mathbb T$. We consider the ideal $I_{\mathcal L',\alpha'}$. It is clear that $I_{\mathcal L,\alpha} \subset I_{\mathcal L',\alpha'}$. Using the description of bases in the respective quotients, derived in the preceding part of this proof, we see that the inclusion is proper. That is, $I_{\mathcal L',\alpha'}/I_{\mathcal L,\alpha}$ is a nontrivial $\mathcal A$-invariant subspace in $V_{\mathcal L, \alpha}$, showing that $V_{\mathcal L, \alpha}$ is not simple.

    Now suppose that $L$ is Lagrangian. We show that $V_{\mathcal L,\alpha}$ is irreducible by demonstrating that every nonzero vector in $V_{\mathcal L,\alpha}$ is cyclic. Let $v$ be as in \eqref{eq:vector_in_rep}, with $n \neq 0$, so that $v$ represents a nonzero element of $V_{\mathcal L,\alpha}$. Acting with $T_l$, $l \in L$, we get the vector
    \begin{equation}
        \alpha(l) \sum_{i=1}^n c_i \omega(l, a_i) T_{a_i} \in \mathcal A. 
    \end{equation}
    If $n \geq 2$, the assumption on $a_1,\dots,a_n$ and that $L$ is Lagrangian imply that for some $l \in L$ not all coefficients $\omega(l,a_i)$ are equal. Then some nonzero linear combination of $v$ and $T_l v$ can be expressed in terms of a proper subset of $\{ T_{a_1}, \dots , T_{a_n } \}$. In other words, the submodule of $V_{\mathcal L, \alpha}$ generated by $v$ contains an element
    \begin{equation}
        v' = \sum_{i=1}^{n'} c_i' T_{a_i'},
    \end{equation}
    where $0 < n'<n$, all $c_i'$ are nonzero, and $a_i'$ represent distinct classes in $A/L$. Repeating this step inductively we find that the submodule generated by $v$ contains an element $T_a$. This element generates $V_{\mathcal L,\alpha}$.
    \end{proof}

    Now let $L \subset A$ be Lagrangian. As one can see from the proof of Proposition~\ref{prop:stab_rep_irreducible}, the~joint eigenspace of operators $T_l$ on $V_{\mathcal L,\alpha}$ to eigenvalues $\alpha(l)$ is one-dimensional and spanned by $[1]$, the class of $1$. If $V$ is any $\mathcal A$-module and $v \in V$ a vector satisfying \eqref{eq:stabilizer_conditions}, there exists a unique module homomorphism $V_{\mathcal L,\alpha} \to V$ such that $[1] \mapsto v$. It is injective by the irreducibility of $V_{\mathcal L,\alpha}$.

\section{Metric groups} \label{app:metric_group}

Let $k$ be an algebraically closed field of characteristic zero, for example $k=\mathbb{C}$. Let $(E, +)$ be a finite abelian group. 
\begin{defn}
    A quadratic form on $E$ is a map $q: E\rightarrow k^\times$  such that $q(a)=q(-a)$ and the symmetric function
    \begin{equation}
        b(a_1, a_2):= \frac{q(a_1+a_2)}{q(a_1)q(a_2)}
        \label{eq:b_from_q}
    \end{equation}
    is a bilinear form, i.e., 
    \begin{equation}
        b(a_1+a_2, a_3)=b(a_1, a_3)b(a_2, a_3).
    \end{equation}
    We say $q$ is non-degenerate if $b(a,\cdot)=1$ implies $a=0$. 
    
    A finite abelian group equipped with a quadratic form is called a pre-metric group. A pre-metric group $(E,q)$ is called a metric group if the bilinear form associated to $q$ as in \eqref{eq:b_from_q} is non-degenerate.
\end{defn}
\begin{lem}
    For a pre-metric group $(E,q)$ and $n\in \ZZ$,
    \begin{equation}
        q(na)=q(a)^{n^2},
    \end{equation}
    for any $a\in E$.
\end{lem}
\begin{proof}
    First note that
    \begin{equation}
        q(0)^{-1} = \frac{q(a+0)}{q(a)q(0)} = b(a,0)=1,
    \end{equation}
    where the last equality holds because $b$ is bilinear. Hence $q(0)=1$. Moreover,
    \begin{equation}
        \frac{q(a-a)}{q(a)q(-a)}=b(a,a)^{-1},
    \end{equation}
    so $b(a,a)=q(a)^2$. 
    
    For the main claim, we can assume that $n\geq 0$ and proceed by induction. Assuming that $q(na)=q(a)^{n^2}$, we find
    \begin{align}
        q((n+1)a) &= b(na, a)q(na)q(a)\\
            &= q(a)^{n^2+2n+1}. \nonumber
    \end{align}
\end{proof}
\begin{prop}\label{prop: roots_of_unity}
    For any $a\in E$, $q(a)$ is a root of unity in $k^\times$.
\end{prop}
\begin{proof} 
    As $E$ is a finite group, there exists $n\in \ZZ$ with $na=0$. The lemma above implies $1=q(na)=q(a)^{n^2}$. 
\end{proof}
\begin{prop}
    A pre-metric group $(E, q)$ splits into a finite direct sum 
    \begin{equation}
    E\cong \bigoplus_{i\in I} E_i,
    \end{equation}
    such that 
    \begin{enumerate}
        \item $E_i$ is a finite $\ZZ_{p_i^{r_i}}$-module, with distinct prime numbers $\{p_i\}_{i\in I}$, 
        \item $q=\prod_{i\in I}q_i$, where $q_i: E_i \rightarrow k^\times$ is the restriction of $q$ on $E_i,$
    \end{enumerate}
\end{prop}
\begin{proof}
    Part 1 follows from the classification of finite abelian groups. To prove part~2, it suffices to assume $I=\{1, 2\}$. Given $a_1 \in E_1$ and $a_2 \in E_2$, 
    \begin{equation}
       b(a_1, a_2)^{p_1^{r_1}}=b(p_1^{r_1}a_1, a_2) =b(0, a_2) = 1,
    \end{equation} 
    and analogously 
    \begin{equation}
       b(a_1, a_2)^{p_2^{r_2}}=b(a_1, p_2^{r_2}a_2) =b(a_1, 0) = 1.
    \end{equation} 
    This implies $b(a_1, a_2)=1$, hence $q(a_1)q(a_2)=q(a_1+a_2)$.
\end{proof}
\begin{remark} \label{rmk: additive_q_convention}
    It is sometimes convenient to have $q$ valued in some $\ZZ_n$ (which is isomorphic to $n$th roots of unity in $\mathbb C$ as an abelian group) in order to utilize additive notation. This is the convention followed in the main text. 
\end{remark}
One motivation for studying (pre-)metric groups lies in their connection to (braided fusion or) modular tensor categories (MTC)~\cite{etingof2015tensor}. The latter are mathematical structures that arise in a variety of fields including quantum algebra, low-dimensional topology, conformal field theory, and topological field theory. Metric groups correspond to a special class of modular tensor categories, which are called pointed.

\begin{thm}[Theorem 8.4.12 in \cite{etingof2015tensor}]
There are equivalences of categories:
\begin{subequations}
\begin{align}
    \{\text{pointed braided fusion categories}\} &\longleftrightarrow \{\text{pre-metric groups}\}, \\
    \{\text{pointed modular tensor categories}\} &\longleftrightarrow \{\text{metric groups}\}.
\end{align}
\end{subequations}
\end{thm}

Modular tensor categories correspond to~(2+1)-dimensional topological quantum field theories via the Reshetikhin–Turaev construction~\cite{reshetikhin1991invariants}, which assigns a TQFT to each MTC. Freed and Teleman~\cite{freed2021gapped} show that such a TQFT can support a gapped boundary if and only if the underlying MTC is equivalent to the Drinfeld center of a~fusion category—that is, when the theory also arises from a~Turaev–Viro model~\cite{turaev1992state}. In the case of pointed MTCs, their condition for supporting a gapped boundary translates to the corresponding metric group being metabolic. This criterion has a natural counterpart in abelian Chern–Simons theory, where gapped boundaries correspond to Lagrangian subgroups of the charge lattice, as studied by Kapustin and Saulina~\cite{kapustin2011topological}.

\begin{defn}
    Let $(E, q)$ be a metric group. A subgroup \(L \subseteq E\) is called \textbf{Lagrangian} if:
\begin{itemize}
  \item \(q|_L = 1\), i.e., \(q(\ell) = 1\) for all \(\ell \in L\) , and
  \item \(L = L^\perp\), where \(L^\perp := \{x \in E \mid b(x, \ell) = 1 \text{ for all } \ell \in L\}\).
\end{itemize}
A metric group \((E, q)\) is said to be \textbf{metabolic} if it contains a Lagrangian subgroup.
\end{defn}
On the other hand, for a 2D (or 2+1 spacetime dimensional) stabilizer code, the condition for supporting a boundary is related to the metabolicity of a quasi-symplectic module \((P_\partial, \Omega_\partial)\). Remarkably, these two notions—arising from lattice models and from topological quantum field theory, respectively—agree perfectly, as~we discuss in this article.
\printbibliography
\end{document}